\documentclass[a4paper,11pt]{article}
\pdfoutput=1 
\usepackage{jcappub} 
\usepackage[T1]{fontenc}
\usepackage{mathtools}
\usepackage{amsthm}
\usepackage{amsmath}
\usepackage{amssymb}
\usepackage{graphicx}
\usepackage{dcolumn}
\usepackage{bm}
\usepackage{mathtools}
\newtheorem{theorem}{Theorem}[section]

\usepackage{graphicx}
\usepackage{float}
\usepackage{tikz} 
\usetikzlibrary{angles}
\usetikzlibrary{automata,positioning}
\usepackage{hyperref}
\usepackage{bm}
\usepackage{longtable}
\title{\boldmath On the physical origin of the radial acceleration relation}
\author[a,1]{Gary Nash,\note{Corresponding author.}}
\affiliation[a,1]{University of Alberta,\\Edmonton, Alberta, Canada, T6G 2R3}
\emailAdd{gnash@ualberta.net\\Present Address: Edmonton, Alberta}

\abstract{The radial acceleration relation (RAR) is a tight correlation between the distribution of baryonic matter in galaxies and their observed dynamics, even in systems where dark matter dominates. This suggests a universal acceleration scale exists in galaxies, which appears to challenge the $\Lambda$CDM standard cosmological model and favor Modified Newtonian dynamics (MOND). However, MOND also deviates from the RAR, so these deviations are not fully explained by the $\Lambda$CDM model or MOND. A theory that explains gravity and the dark sector geometrically, and treats dark matter as a spin-1 particle, is required, such as Modified General Relativity (MGR). It is the only metric-compatible theory of gravity that describes local gravitational energy-momentum by the connection-independent symmetric tensor $\varPhi_{\alpha\beta}$, which completes the total energy-momentum tensor of the Einstein equation. Dark energy is explained by the energy-related component $\varPhi_{00}$, and dark matter by the stress components $\varPhi_{ij}$. Dark matter is extremely difficult to observe because it is fundamentally related to the Lorentzian metric of spacetime itself. The physical origin of the RAR follows directly from the equilibrium between the dark energy and Newtonian forces in the outer part of the rotation curve of galaxies or galaxy clusters, and that between the dark sector forces in the inner part of the rotation curve. It is shown that MGR subsumes MOND, accounts for the gravitational anomaly observed in wide binaries, and exhibits the observed Newtonian behavior of gravity at large scales.}
\begin{document}
\maketitle
	\flushbottom
		
\section{Introduction}
Within the realm of General Relativity (GR), the precise origin of the radial acceleration relation (RAR) is not definitively known and is currently a subject of intense debate in astrophysics. The RAR is an empirical observation \cite{McG} that creates a tight, consistent correlation between the observed acceleration $a_{obs}$ of stars/gas in a galaxy and the acceleration expected from baryonic (ordinary) matter alone $a_{bar}$; it applies to both late time and early time galaxies and subsumes and generalizes several well-known dynamical properties of galaxies, like the Tully-Fisher and Faber-Jackson relations, the baryon-halo conspiracies, and Renzo's rule \cite{Lelli}. The tight correlation between the distribution of baryonic matter in galaxies and their observed dynamics, even in systems where the dark matter dominates, is puzzling in the standard $\Lambda$-cold dark matter ($\Lambda$CDM) model\cite{Julio}.
\par The RAR was inspired by  Modified Newtonian Dynamics (MOND)\cite{Mila,Milb,Milc}. It proposes that Newtonian dynamics holds at high accelerations. However, as acceleration drops below the fundamental constant value  $A_{0}= 1.20\pm 0.02\times10^{-10}\; ms^{-2}$ where Newtonian gravity breaks down, the dynamical mass increasingly exceeds the baryonic mass such that $a_{obs}\propto\sqrt{a_{bar}}$ in the weak-field regime. MOND is a phenomenological theory that dispenses with dark matter. The RAR is empirically described \cite{McG} by the relation obtained from an interpolation function in MOND:
\begin{equation}\label{RAR}
	a_{obs}=\frac{a_{bar}}{1-e^{-\sqrt{\frac{a_{bar}}{A_{0}}}}}.
\end{equation}
\par MOND reproduces the RAR with zero scatter \cite{Lelli}. However, the external field effect in it causes galaxies in strong external fields to deviate from the RAR \cite{Gar}. Moreover, MOND does not follow the strong equivalence principle, so the internal dynamics of satellite galaxies orbiting in the tidal field of their host galaxy experience an external field effect \cite{Julio}. The difficult challenge for MOND is to find a satisfactory covariant theory that reproduces all the successes of General Relativity as well as MOND in the appropriate limits. 
\par Galaxy formation in $\Lambda$CDM proceeds in a highly stochastic and complicated manner that makes it difficult to explain the RAR. However, hydrodynamical simulations within the $\Lambda$CDM paradigm \cite{Paran} demonstrate that the RAR naturally arises from baryonic physics and galaxy formation feedback. Consequently, the RAR does not falsify the standard dark matter model.  
\par The RAR subsumes the baryonic Tully-Fisher relation (BTFR) $M=Av^{x}$. It relates the total baryonic mass $M$ of a galaxy to its asymptotically flat rotational velocity $v$, and has been well established from astronomical data\cite{Schom}. The standard $\Lambda$CDM model predicts an exponent between three and four because the BTFR must emerge from the complex process of galaxy formation, and is expected to show significant intrinsic scatter  \cite{Goddy}. However, the BTFR exhibits an exceptionally low intrinsic scatter, which poses a significant challenge for the $\Lambda$CDM cosmological model \cite{LMS}. In contrast, MOND predicts a BTFR slope of exactly four with no scatter. Nevertheless, the large departure of super spirals from a power-law BTFR with a slope of 4.0 is inconsistent with MOND\cite{Ogle}.
\par Recent studies of wide binary stars reveal an unambiguous and strong breakdown of Newtonian dynamics and GR at low accelerations  $ <10^{-9}m s^{-2}$ \cite{Chae}. What is even more surprising is that the trend and magnitude of the gravitational anomaly agree with what the A-QUAdratic Lagrangian theory (AQUAL) predicts \cite{Beck}. 
Given the inferred local volume density of dark matter, its total content within a wide binary orbit is assumed to be negligible in comparison to the masses of the stars themselves. Thus, based on this assumption, any gravitational anomaly of the type encountered at galactic scales and beyond, found in wide binary stars, cannot comfortably be ascribed to the presence of dark matter\cite{Hern1}. However, this analysis did not include dark energy and the assumption that dark matter is negligible is questionable. This gravitational anomaly presents immense implications for astrophysics, cosmology, and fundamental physics. 
\par The crux of these tensions or anomalies is whether or not dark matter exists; the $\Lambda$CDM standard model fundamentally depends on it whereas MOND and AQUAL dispense with dark matter. To understand the physical origin of the RAR, we must first comprehend the nature of dark matter and dark energy, which should also elucidate the wide binaries anomaly and the large scale nature of gravity.
\par It is well known that GR lacks a symmetric tensor that describes local gravitational energy-momentum; the Einstein equivalence principle forbids it \cite{MWT}. That foundational issue is resolved in Modified General Relativity (MGR) \cite{N,ND,NP}, which fundamentally introduces a connection-independent symmetric tensor $\varPhi_{\alpha\beta}$ describing local gravitational energy-momentum. $\varPhi_{\alpha\beta}$ is an integral part of the total energy-momentum tensor in MGR. Its energy component $\varPhi_{00}$ and the stress components $\varPhi_{ij}\;i,j=1,2,3$ describe dark energy and dark matter, respectively. The dynamics between dark energy and ordinary matter fundamentally replace the empirical relation (\ref{RAR}) with the derived expression (\ref{aout}). How $\varPhi_{\alpha\beta}$ enters the Einstein equation, both geometrically and dynamically, and describes dark energy and dark matter is now reviewed.
\subsection{Overview of Modified General Relativity}
To discover $\varPhi_{\alpha\beta}$, it is important to understand the fundamental structure of all Lorentzian metrics in terms of the line element vector field. Every noncompact paracompact differentiable manifold $ \mathcal{M} $ admits a continuous non-zero vector field and thus a continuous line element field \cite{Markus}. There exists a differentiable regular (everywhere nonvanishing) vector field $\bm{X}$ collinear to a differentiable nonvanishing unit vector $ \bm{u} $ by $\bm{X}=f\bm{u}$ where $ f\neq0$ is the scalar magnitude of $ \bm{X}$. The line element (or direction) field $(\bm{X},-\bm{X})$ is defined as an assignment of a pair of equal and opposite vectors at each point of $\mathcal{M}$. It is a one-dimensional vector subspace of the tangent space on $ \mathcal{M} $, or equivalently, a line subbundle $l$ of the tangent bundle on $\mathcal{M}$. 
\par This leads to a technical issue that needs to be addressed. While $\mathcal{M}$ admits a nonvanishing vector field, it does so as a vector field with isolated zeros \cite{Markus}, which can be swept out to infinity. However, the resulting nonvanishing vector field is unsatisfactory because it may not be bounded \cite{kato}. In MGR, the line subbundle $l$ is defined from the vectors (\ref{u}) obtained from the variation of (\ref{S}) with respect to the unit vector $u^{\nu}$. Boundedness is guaranteed from the structure of the covectors and vectors in $l$ by demanding $u_{\alpha}$, $u^{\alpha}$, $\partial_{\alpha}f$, $\partial^{\alpha}f$, and $f$ to be bounded. $\Phi$ is bounded because it satisfies (\ref{intPhi}). Thus, the line subbundle $l$ is bounded and is constructed from regular vector fields. This ensures the boundedness of solutions to the Einstein equation in MGR that explicitly depend on the line element vectors in $l$. 
\par A Lorentzian metric $g_{\alpha\beta}$ with a +2 signature is related to a Riemannian metric by
\begin{equation}\label{gab}
	g_{\alpha\beta}=g_{\alpha\beta}^{+}-2u_{\alpha}u_{\beta}
\end{equation}where the Riemannian metric $g^{+}_{\alpha\beta}$ is independent of the covectors $u_{\alpha}$ and the unit vectors satisfy $g_{\alpha\beta}u^{\alpha}u^{\beta}=-1$, and $g^{+}_{\alpha\beta}u^{\alpha}u^{\beta}=1$. A noncompact paracompact manifold admits a Lorentzian metric $ g_{\alpha\beta} $ if and only if it admits a line element field \cite{CB,Hawk}. The regular vectors in the line element field are not introduced arbitrarily as absolute vectors in a Lorentzian spacetime $(\mathcal{M},g_{\alpha\beta})$ because it does not exist without them. 
\par The line element vectors at each point of a Lorentzian spacetime can be divided into three classes depending on whether $X_{\beta}X^{\beta}$ is negative, positive or zero: timelike, spacelike, and null, respectively. The null vectors form the familiar double cones in the tangent space on $\mathcal{M}$, which separates the timelike vectors from the spacelike
vectors. A Lorentzian metric has a directional characteristic that is described by the line element vectors associated with the metric.
\par MGR selects a particular bounded $l$ consisting of those unit covectors satisfying (\ref{u}) constrained by (\ref{nu}) from the myriad of covectors that could belong to that $l$. MGR depends explicitly on the Riemannian metric on $\mathcal{M}$ and the collinear vectors \textbf{u} and $\textbf{X}$ in $l$ from the introduction of $\varPhi_{\alpha\beta}$ into the complete energy-momentum tensor of the Einstein equation. The explicit dependence on the vectors in the bounded line subbundle of MGR provides it with the extra freedom required to describe dark energy, and dark matter without introducing a dark matter profile; higher than second-order derivatives or higher than four-dimensional spacetimes are not necessary.
\par Whereas GR is constructed in the tangent bundle on $\mathcal{M}$ that contains a regular timelike vector $\textbf{X}$ and a Riemannian metric. GR describes the gravitation of ordinary matter with its definitive dependence on the Riemannian metric; GR exists on a Lorentzian manifold but does not depend explicitly on $\textbf{X}$.  
\par The geometrical development of $\varPhi_{\alpha\beta}$ follows by re-introducing Einstein's postulate \cite{EinGros,Ein16} of a \emph{total} 
energy-momentum tensor $T_{\alpha\beta}$, which must contain the matter energy-momentum tensor $\tilde{T}_{\alpha\beta}$ and another symmetric tensor that geometrically describes the energy-momentum of the gravitational field. The local conservation of total energy-momentum then requires $\nabla^{\alpha}T_{\alpha\beta}=0$. $\tilde{T}_{\alpha\beta}$ is non-divergenceless and can be orthogonally decomposed by the Orthogonal Decomposition Theorem (ODT) in Appendix A. That exposes the Lovelock tensors \cite{Love} $g_{\alpha\beta}$ and $G_{\alpha\beta}$ in a four-dimensional spacetime, and introduces a new symmetric tensor that represents the energy-momentum of the gravitational field, $\varPhi_{\alpha\beta}$. The ODT and Lovelock's theorem fundamentally and geometrically generate the modified Einstein equation  
\begin{equation}\label{MEQ}
	\Lambda g_{\alpha\beta}+G_{\alpha\beta}+\varPhi_{\alpha\beta}=\kappa\tilde{T}_{\alpha\beta}
\end{equation}
where $\kappa:=\frac{8\pi G}{c^{4}}$, $\Lambda$ is the cosmological constant, and
\begin{equation}\label{Phiab}
	\begin{split}
		\varPhi_{\alpha\beta}=\frac{1}{2}\mathcal{L}_{X}g^{+}_{\alpha\beta}=\frac{1}{2}\mathcal{L}_{X}g_{\alpha\beta}+\mathcal{L}_{X}(u_{\alpha}u_{\beta}),
	\end{split}
\end{equation} which can be expressed as
\begin{equation}\label{Phiabex}
	\varPhi_{\alpha\beta}=\frac{1}{2}(\nabla_{\alpha}X_{\beta}+\nabla_{\beta}X_{\alpha})+u^{\lambda}(u_{\alpha}\nabla_{\beta}X_{\lambda}+u_{\beta}\nabla_{\alpha}X_{\lambda})
\end{equation}in an affine parameterization. The vectors \textbf{u} and \textbf{X} in $\varPhi_{\alpha\beta}$ are directly related to a given Lorentzian metric and are not introduced arbitrarily. $\varPhi_{\alpha\beta}$ is the change of the Riemannian metric $g^{+}_{\alpha\beta}$, that specifically belongs to a particular Lorentzian metric $g_{\alpha\beta}$, along the flow of $\bm{X}$. 
\par $\varPhi_{\alpha\beta}  $ vanishes if and only if $ X^{\mu} $ is a Killing vector. However, in general, there are no Killing vector fields unless a particular symmetry is involved.  
\par $\varPhi_{\alpha\beta}$ is constructed from the Lie derivatives in (\ref{Phiab}), which do not depend on the connection $\nabla$ that defines the covariant derivative on the manifold. It is straightforward to show from the definition of a covariant derivative that the connection coefficients (Christoffel symbols) cancel out of the Lie derivative of any tensor leaving partial derivatives in place of the covariant derivatives. Any covariant expression obtained from the Lie derivative of a tensor is a tensor of the same rank that is equivalently expressed in terms of partial derivatives. The connection coefficients represent the gravitational field in a metric theory of gravity, and they vanish during free fall.  Thus, $\varPhi_{\alpha\beta}$ is invariant during free fall, and there is no conflict with the Einstein equivalence principle. 
\par Equation (\ref{MEQ}) can be obtained dynamically from the action functional $ S=S^{F}+S^{EH}+S^{G} $, which consists of the action for all ordinary matter fields $S^{F}$, the Einstein-Hilbert action of GR, $S^{EH}$, and the action for the energy-momentum of the gravitational field, $S^{G}$:
\begin{equation}\label{S}
	\begin{split}
		S=\int L^{F}( A^{\beta},\nabla^{\alpha} A^{\beta},...,g^{\alpha\beta})\sqrt{-g}d^{4}x
		+\frac{c^{3}}{16\pi G} \int (R-2\Lambda)\sqrt{-g}d^{4}x\\-\frac{c^{3}}{16\pi G}\int \varPhi_{\alpha\beta} g^{\alpha\beta}\sqrt{-g} d^{4}x
	\end{split}
\end{equation}
where $L^{F}  $ is the Lagrangian of the ordinary matter fields $A^{\beta}$. The details of the variations of $S$ with respect to the variables $g^{\alpha\beta}$, $ u^{\nu}$ and $f $ are given in Appendix B. In particular, variation of $S$ with respect to the inverse metric generates (\ref{MEQ}), the term $-X^{\lambda}\nabla_{\lambda}(u^{\alpha}u^{\beta})$, and the constraint
\begin{equation}\label{uab}
	\nabla_{\alpha}(u^{\alpha}u^{\beta})=0
\end{equation}from which the global constraint
\begin{equation}\label{intPhi}
	\int\Phi\sqrt{-g}d^{4}x=\int\nabla_{\alpha}X_{\beta}(g^{\alpha\beta}+2u^{\alpha}u^{\beta})=0
\end{equation} follows after integration by parts. The negative sign in $-X^{\lambda}\nabla_{\lambda}(u^{\alpha}u^{\beta})$ demands that this term must vanish, which is consistent with expressing it in an affine parameterization as in (\ref{Phiabex}).
\par As shown in \cite{ND}, the action $
S^{EHG}=\frac{c^{3}}{16\pi G}\int (R-\Phi)\sqrt{-g}d^{4}x $ that follows from (\ref{intPhi}) generates the modified Einstein equation with no cosmological constant. If $ \Phi$ is set to the constant 2$\Lambda $, the Einstein equation with a cosmological constant is obtained accordingly, which contradicts (\ref{intPhi}). Thus, $ \Phi $ dynamically replaces the cosmological constant and the Einstein equation with a \emph{complete} energy-momentum tensor is
\begin{equation}\label{ME}
	G_{\alpha\beta}=\frac{8\pi G}{c^{4}}T_{\alpha\beta}
\end{equation} with
\begin{equation}\label{T}
	T_{\alpha\beta}=\tilde{T}_{\alpha\beta}-\frac{c^{4}}{8\pi G}\varPhi_{\alpha\beta}. 
\end{equation}
The local conservation law, $\nabla^{\alpha}T_{\alpha\beta}$ = 0, follows from the diffeomorphic invariance of MGR and is consistent with  Einstein's original postulate.
\par Variation of the action functional $S$ with respect to $u^{\nu}$ generates
\begin{equation}\label{u}
	u_{\nu}=\frac{\partial_{\nu}f}{\Phi}
\end{equation}and variation with respect to $f$, the magnitude of $X^{\alpha}$, yields the constraint
\begin{equation}\label{nu}
	\nabla_{\alpha}u^{\alpha}=0.
\end{equation}
\par Exposing the line element vector field from $\varPhi_{\alpha\beta}$ leads to new solutions to the Einstein equation from (\ref{ME}) that depend on the line element covectors in the Lorentzian metric (\ref{gab}). 
With $\textbf{u}$ in an orthonormal basis,
\begin{equation}\label{X0}
	\varPhi_{00}=(1+2u_{0}u^{0})\nabla_{0}X_{0}
\end{equation}and 
\begin{equation}\label{Xi}
	\varPhi_{ii}=(1+2u_{i}u^{i})\nabla_{i}X_{i}\;\;\text{no sum on i},\;\;i=1,2,3.
\end{equation}  The introduction of $\varPhi_{\alpha\beta}$ into GR yields all results of GR plus additional features attributed to dark matter and dark energy. The local gravitational energy-momentum tensor $\varPhi_{\alpha\beta}$ describes the dark Universe geometrically. Its absence in GR is why it has been difficult to understand what dark energy is and why a dark matter profile must be introduced into GR to account for the invisible matter in galaxies known as dark matter. 
\subsubsection{Testing MGR against astronomical data}
MGR has been tested against existing data in \cite{ND}. For example:
\begin{itemize}
	\item the power four Tully-Fisher relation follows naturally from the solution of (\ref{ME}) in a spherically symmetric spacetime free of ordinary matter;
	\item  the calculation of the total advance of the perihelion of Mercury due to GR, dark matter from the line element vectors, and the quadrupole moment of the Sun yielded 43.012 arcsec/century, which compared excellently to the measured result of $43.0115\pm0.0085$ arcsec/century \cite{Pir};
	\item the strong lensing from galactic cluster SDSS J0900+2234 with an Einstein ring calculated in MGR yields 13.06 arcsec that compares well to the $10.5\pm 3.6$ arcsec value \cite{Sloan} calculated in GR with the NFW dark matter profile;
	\item  the Shapiro time delay for the Crab pulsar calculated in MGR is 7.63 days. This result can be compared to three calculations of the Shapiro time delay for the Crab pulsar quoted in \cite{Desai} as (5.14-15.42) d, 1.98 d and 3.84 d. The MGR result is larger than the last two noted results that are determined in GR with the NFW dark matter profile, commensurate with the fact that the NFW dark matter profile is not valid in the interval [5,30] kpc from the galactic center of the Milky Way;
	\item As first discussed in \cite{ND} and noted in Section 3, MGR predicts the time dependence of dark energy that has been confirmed to $\sim4\sigma $.
	\item  Gravitational lensing of the Bullet Cluster was calculated from MGR based on the mass estimates in \cite{Para}. The main cluster generated 67.04 arcsec with the GR term contributing 7.65 arcsec to that result because GR itself only describes ordinary matter. The smaller Bullet cluster lenses 50.71 arcsec with a contribution of 6.58 arcsec from the GR term. Thus, ordinary matter accounts for 11.4\% and 13.0\% of the total gravitating mass of the main and Bullet clusters, respectively, as compared to 15.6\% for each cluster according to the $\Lambda$CDM model.
\end{itemize}
\subsubsection{The geometrical nature of \texorpdfstring{$\varPhi_{\alpha\beta}$}{varPhialphabeta} from the Lie derivative}
It is now shown from the Lie derivative how dark matter changes, independently of ordinary matter, along the flow of the line element vector of the corresponding Lorentzian metric. This is the geometrical nature of dark matter, which nicely describes how the dark matter of the smaller Bullet cluster flows unimpeded along its line element vector through the gas of the main cluster. After the collision of the clusters, we see the hot plasma from the heated gas of the clusters surrounded by the dark matter haloes of each cluster in the merged Bullet Cluster.
\par  Any action functional with a scalar Lagrangian constructed from the tensor fields of the Lagrangian is invariant under the group Diff($\mathcal{M}$) of diffeomorphisms. Since $\varPhi_{\alpha\beta}$ is obtained from $S^{G}$, the Lie derivative of the Lorentzian metric and its associated unit covectors have a geometrical interpretation in terms of a family of diffeomorphisms. Given a diffeomorphism $ \phi: \mathcal{M}\longrightarrow \mathcal{M} $, the Lie derivative of the metric is constructed from the pullback $\phi_{t\ast}$ of the metric under Diff($\mathcal{M}$):
\begin{equation}\label{Lie}
	\mathcal{L}_{X} g_{\alpha\beta}= \lim_{t \to 0} \{\frac{\phi_{t\star}[g_{\alpha\beta}(\phi_{t}(p))]-g_{\alpha\beta}(p)}{t}\}.
\end{equation} $\phi_{t} $ is the flow down the integral curves defined as those curves $x^{\beta}(t)$, which solve $ X^{\beta}=\dfrac{dx^{\beta}}{dt} $ for the line element vector field $X^{\beta}$ where $ t $ defines a one-parameter family of diffeomorphisms. The Lie derivative of the metric tells us how fast the metric changes as it moves along the integral curves. Similarly for the $u_{\alpha}u_{\beta}$ part of $\varPhi_{\alpha\beta}$. From (\ref{Phiab}), $\varPhi_{ij}=\frac{1}{2}\mathcal{L}_{X}g_{ij}+\mathcal{L}_{X}(u_{i}u_{j})$, so dark matter is determined from the rate of change of the stress components of the metric and its corresponding unit covectors along the flow of its line element vector $\bm{X}$. $\varPhi_{ij}$ depends only on the line element covectors and vectors in (\ref{Phiabex}); dark matter is independent of ordinary matter and is intrinsically related to a particular Lorentzian spacetime; similarly for dark energy and $\varPhi_{00}$.
\subsubsection{Dark energy and dark matter}
A spatially maximally symmetric (homogeneous and isotropic) spacetime is described by the Friedmann-Lema\^{i}tre-Robertson-Walker (FLRW) metric
\begin{equation}\label{FLRW}
	ds^{2}=-c^{2}dt^{2}+a(t)^{2}[\frac{1}{1-k r^{2}}dr^{2}+r^{2}(d\theta^{2}+{sin^{2}\theta} d\varphi^{2})]
\end{equation} where $a(t)$ is the dimensionless cosmological scale factor, which satisfies $a>0 $ after the Big Bang at $ t=0 $.  The parameter $k$ of dimension $L^{-2}$ is used to describe a particular spatial geometry; $k=1,0,-1$ describes a closed, flat, or open space with constant curvature, respectively, and $r$ is the comoving radius of dimension $L^{1}$. \par A spatially maximally symmetric second rank (0,2) tensor $B_{\alpha\beta}$ has the components 
\begin{equation}\label{Maxsym}
B_{00}=\alpha(t),\enspace B_{0j}=0,\enspace B_{ij}=\beta(t)g_{ij}\;\;i,j=1,2,3
\end{equation} where $\alpha(t)$ and $\beta(t)$ are arbitrary functions of time. $\varPhi_{\alpha\beta}$ cannot be expressed as a spatially maximally symmetric tensor because that symmetry would cause it to vanish. However, we can still explore it in the FLRW metric. The maximal spatial symmetry requires $R_{1}^{1}=R_{2}^{2}=R_{3}^{3}$, which demands $\varPhi^{1}_{1}=\varPhi^{2}_{2}=\varPhi^{3}_{3}$ from (\ref{Rab}). $\varPhi_{00}=\varPhi_{00}(t)$ is consistent with (\ref{Maxsym}) and with (\ref{1Phi00}) restricted to a function of time. Ordinary matter is described in the usual spatially maximally symmetric form with $\tilde{T}_{00}=\varrho c^{2}$, $\tilde{T}_{ij}=pg_{ij}$, and $\tilde{T}=-\varrho c^{2}+3p$ where $\varrho$ is its mass density and $p$ is its pressure. This allows for a direct comparison of the field equations to the usual Friedmann equations.
\par  To obtain the field equations, we use the trace of the Einstein equation with the complete energy-momentum tensor $-\frac{8\pi G}{c^{4}}\tilde{T}-R+\Phi=0$ to rewrite it as
\begin{equation}\label{Rab}
R_{\alpha\beta}=\frac{8\pi G}{c^{4}} (\tilde{T}_{\alpha\beta}-\frac{1}{2}g_{\alpha\beta}\tilde{T})+\frac{1}{2}g_{\alpha\beta}\Phi -\varPhi_{\alpha\beta}
\end{equation} from which we obtain
\begin{equation}\label{F1}
\frac{\ddot{a}}{a}=-\frac{4\pi G}{3}(\varrho+\frac{3p}{c^{2}})+\frac{c^{2}}{6}(2\varPhi_{00}+\Phi)
\end{equation} from the $R_{00}$ component. The $R_{11}$ component gives
\begin{equation}\label{frw2}
\frac{\ddot{a}}{a}+\frac{2\dot{a}^{2}}{a^{2}}+\frac{2kc^{2}}{a^{2}}=4\pi G(\varrho-\frac{p}{c^{2}})+\frac{c^{2}}{2}(\Phi-2\varPhi_{1}^{1})
\end{equation} and the local conservation law $\nabla_{\alpha}T^{\alpha\beta}=0$  yields the conservation equation 
\begin{equation}\label{CE}
\dot{\varrho}-\frac{c^{2}}{8\pi G}\dot{\varPhi_{00}}=-\frac{3\dot{a}}{a}(\varrho+\frac{p}{c^{2}})+\frac{3c^{2}\dot{a}}{8\pi Ga}(\varPhi_{00}+\varPhi_{1}^{1}).
\end{equation}
Inserting (\ref{F1}) into (\ref{frw2}) produces a simpler equation
\begin{equation}\label{F2}
(\frac{\dot{a}}{a})^{2}=\frac{8\pi G\varrho}{3}-\frac{c^{2}\varPhi_{00}}{3}-\frac{kc^{2}}{a^{2}}. 
\end{equation} 
Equations (\ref{F1}) and (\ref{F2}) are the modified Friedmann equations in the FLRW metric in cosmic time.
\par In the FLRW metric, $R_{ij}=(\frac{\ddot{a}}{ac^{2}}+\frac{2\dot{a}^{2}}{a^{2}c^{2}}+\frac{2\kappa }{a^{2}})g_{ij}$, $R_{1}^{1}=R_{2}^{2}=R_{3}^{3}$ and it follows from (\ref{Rab}) that
\begin{equation}\label{var123}
	\varPhi_{1}^{1}=\varPhi_{2}^{2}=\varPhi_{3}^{3}.
\end{equation}
\par With the Hubble parameter defined as $ H:=\frac{\dot{a}}{a} $, 
\begin{equation}\label{dH}
\frac{\ddot{a}}{a}=\dot{H}+H^{2}, 	
\end{equation} 
and with (\ref{F1}), (\ref{CE}), and (\ref{F2})
\begin{equation}\label{dPhi00}
\frac{d\varPhi_{00}}{da}=\frac{8\pi G}{c^{2}}\frac{d\varrho}{da}+\frac{6\kappa}{a^{3}}-\frac{6\dot{H}}{ac^{2}},
\end{equation} which has the solution
\begin{equation}\label{1Phi00}
\begin{split}
	\varPhi_{00}=\Lambda+\frac{8\pi G\varrho}{c^{2}}-\frac{3\kappa}{a^{2}}-\frac{6}{c^{2}}\int \frac{\dot{H}}{a}da
\end{split}	
\end{equation}for the local gravitational energy. $\Lambda$ is an integration constant that is (or is proportional to) the cosmological constant, which appears immediately after the Big Bang in the primordial epoch when there is no ordinary matter.
\par Dark energy is believed to be responsible for the observed accelerated expansion of the current epoch of the Universe. If $2\varPhi_{00}+\Phi$ is positive and satisfies $2\varPhi_{00}+\Phi>\frac{8\pi G}{c^{2}}(\varrho+\frac{3p}{c^{2}})$ in (\ref{F1}), the Universe accelerates. Since $2\varPhi_{00}+\Phi=\varPhi_{00}+\varPhi_{i}^{i} $, it represents the natural splitting of $\varPhi_{\alpha\beta}$ into its energy and the spatial trace of its stress components, with the understanding that $\frac{c^{4}\varPhi_{\alpha\beta}}{8\pi G}$ is an energy-momentum density. This suggests dark energy can be described by a positive $\varPhi_{00}$ and dark matter by $\varPhi_{1}^{1}$ because the stress components $\varPhi_{ij} $ in the FLRW metric are related to its spatial trace by $\varPhi_{i}^{i}=g^{ij}\varPhi_{ij}=3\varPhi_{1}^{1} $. Then, if the dark energy is large enough relative to the dark matter and the ordinary matter and pressure in a particular epoch of the Universe, it will cause that epoch to accelerate. Thus, the local gravitational energy-momentum tensor $\varPhi_{\alpha\beta}$ provides an explanation of dark energy and dark matter; dark energy is the local gravitational energy $\varPhi_{00}$ and dark matter is described by $\varPhi_{ij}$.  
\subsubsection{Uniqueness relative to all metric-compatible theories of gravity}
As discussed in \cite{NP}, there are many metric compatible modifications of GR obtained by adding a symmetric tensor $h_{\alpha\beta}$ to the ordinary matter tensor so that the Einstein equation holds: 
\begin{equation}\label{GMR}
	G_{\alpha\beta}+\Lambda g_{\alpha\beta}=k(\tilde{T}_{\alpha\beta}+h_{\alpha\beta}) 	
\end{equation}
where $h_{\alpha\beta}$ consists of a linear sum of divergenceless non-Lovelock tensors and the Lovelock tensors in a spacetime of greater than four dimensions. For example, the non-trivial Codazzi tensors $\xi_{\beta\lambda} $ satisfy  $\nabla_{\alpha}\xi_{\beta\lambda}=\nabla_{\beta}\xi_{\alpha\lambda}$. The tensor $H_{\alpha\beta}=\xi_{\alpha\beta}-g_{\alpha\beta}\xi_{\mu}^{\mu}$ with $\nabla^{\alpha}\xi_{\alpha\beta}=g_{\alpha\beta}\nabla^{\alpha}\xi_{\mu}^{\mu} $ can be constructed from the Codazzi tensors. In GR, $\tilde{T}_{\alpha\beta}$ is divergenceless and that property of $H_{\alpha\beta} $ ensures $\nabla^{\alpha}(\tilde{T}_{\alpha\beta}+H_{\alpha\beta})=0.$ Many theories of gravity can be described from a suitable choice of the Codazzi tensors \cite{Mant} in the FLRW metric, including $f(R)$, Gauss-Bonnet $f(G)$, teleparallel $f(T)$, Lovelock gravity, Einsteinian cubic $f(P)$, and Conformal Killing gravity (CKT). Horndeski and tensor-vector-scalar modifications of GR have the same form as (\ref{GMR}) but with additional equations of motion for the added tensor, vector and scalar fields.\par However, (\ref{GMR}) does not contain $\varPhi_{\alpha\beta}$. The ODT guarantees that it is the \emph{only} tensor orthogonal to \emph{all} symmetric divergenceless tensors with $\tilde{T}_{\alpha\beta}\neq0$, including $h_{\alpha\beta}$ and the Lovelock tensors. No divergenceless non-Lovelock tensor or higher dimensional Lovelock tensor is connection-independent. That property of $\varPhi_{\alpha\beta}$ is unique relative to all other modifications of GR with a metric compatible connection, which sets MGR apart from competing theories of gravity.
 \subsubsection{The link between gravity and quantum theory}
MGR links gravity to quantum theory \cite{NP}. The geometrical unification of gravity described by MGR and quantum field theory (QFT) was established by locally embedding a background-independent four-dimensional Lorentzian manifold where gravity lives with all known elementary particles into a flat ten-dimensional Minkowskian manifold where QFT resides.  Dark matter can be described geometrically and as a spin-1 particle in quantum theory by the line element covector $X_{\beta}$, which satisfies both the Proca field equation with a current $J_{\beta}$, $\nabla^{\alpha}K_{\alpha\beta}=k^{2}X_{\beta}+J_{\beta}$, and the spin-1 wave equation $\square X_{\beta}=k^{2}X_{\beta}$ 
where $\square:=\nabla_{\alpha}\nabla^{\alpha}$ and $K_{\alpha\beta}=\nabla_{\alpha}X_{\beta}-\nabla_{\beta}X_{\alpha}$. The Lie derivative of the Lorentzian metric is the symmetric part of the symmetrized wave equation that links gravity and quantum theory for a particular Lorentzian metric: $\square X_{\beta}=\frac{1}{2}\nabla^{\alpha}(\mathcal{L}_{X}g_{\alpha\beta}+K_{\alpha\beta})$ and $\mathcal{L}_{X}g_{\alpha\beta} $ in the spin-1 wave equation is identical to that in $\varPhi_{\alpha\beta}$ for a given Lorentzian metric. Since $X_{\beta}$ intrinsically belongs to a Lorentzian metric, it is not surprising that a dark matter particle characterized by $X^{\beta}$ is yet to be discovered. A Lorentzian spacetime contains a myriad of line element vectors with quantum field properties, but that does not guarantee detectable dark matter particles exist; particle-like excitations in spacetime itself are unlikely to be detected unless some extraordinary event locally creates them large and long enough to be observed. Nevertheless, dark matter has a fundamental geometrical description that does not depend on the existence of any particle whatsoever.
\subsection*{Summary}
Local gravitational energy-momentum described by $\varPhi_{\alpha\beta}$ explains the nature of dark energy and dark matter. The energy component $\varPhi_{00}$ is dark energy, and the stress components $\varPhi_{ij}\:\:i,j=1,2,3$ describe dark matter. The line element vector field exists at all points in spacetime, which implies $\varPhi_{\alpha\beta}$ and therefore dark energy and dark matter exist everywhere and are spacetime dependent. Thus, their gravitational effects are spacetime dependent. It will be shown that MGR unlocks the quandary of dark matter between the standard $\Lambda$CDM model and MOND; it is a covariant theory that reproduces all the successes of General Relativity as well as MOND in the appropriate equilibrium conditions. Although dark matter is described as a spin-1 particle that may never actually be observed, the geometry of spacetime rules gravity, and it will be shown that MGR encompasses MOND.
\par This paper now proceeds as follows: In Section 2, the components of $\varPhi_{\alpha\beta}$ are explicitly presented and an updated solution to the Einstein equation in a spherically symmetric spacetime is presented. The relationship of the parameters involved is established from the gravitational energy density. The velocity rotation curves involving dark energy and dark matter, the RAR, and the baryonic and power-four Tully-Fisher relations are derived from the modified Newtonian force of MGR in Section 3. The relation of MGR to and issues with MOND, and extended flat rotation curves are discussed. Wide binaries are investigated in Section 4, and the behavior of gravity in MGR at large scales is discussed in Section 5.
\section{The spherical symmetric solution to the Einstein equation in the vacuum}
In a region of spacetime free of all ordinary matter, $\tilde{T}_{\alpha\beta}=0$ and 
\begin{equation}\label{EF}
	G_{\alpha\beta}+\varPhi_{\alpha\beta}=0.
\end{equation} An updated spherically symmetric solution to (\ref{EF}) is sought with the metric 
\begin{equation}
		ds^{2}=-e^{\nu} c^{2}dt^{2}+e^{\lambda}dr^{2}+r^{2}(d\theta^{2}+sin^{2}\theta d\varphi^{2})
\end{equation} where $ \nu $ and $ \lambda $ are functions of $r$ and $t$. Since the metric is spherically symmetric, there exist three Killing vectors associated with the spherical symmetry of a three-dimensional spatial rotation. $ \varPhi_{\alpha\beta} $ vanishes when $ X^{\beta} $ is a Killing vector and (\ref{EF}) reduces to $ R_{\alpha\beta}=0 $. It follows that $ \nu=-\lambda $ holds in MGR as a direct result of the spherical symmetry invoked to solve (\ref{EF}). 
 \par  Static solutions to (\ref{EF}) requires $ \partial_{0}X_{\alpha}=0 $ and from the metric $ \partial_{0}\lambda=0$. It follows that $ u_{\alpha}\partial_{0}f+f\partial_{0}u_{\alpha}=0 $ and 
 \begin{equation}\label{uu1}
 	u_{\alpha}u^{\alpha}=-1
 \end{equation}
  holds. 
  \par With $X_{3}=0$ and $\mu:=(1+2u_{0}u^{0})$, the components of $ \varPhi_{\alpha\beta} $ to be considered are then: 
  \begin{equation}\label{Phi00}
  	\varPhi_{00}=\frac{\mu}{2}e^{-2\lambda}\lambda^{\prime} X_{1}
  \end{equation}
  \begin{equation}
  	\varPhi_{11}=(1+2u_{1}u^{1} )({X_{1}}^{\prime}-\frac{1}{2}\lambda^{\prime}X_{1}),
  \end{equation} 
  \begin{equation}
  	\varPhi_{22}=(1+2u_{2}u^{2} )(\partial_{2}X_{2}+re^{-\lambda} X_{1}),  
  \end{equation} 
  \begin{equation}
  	\varPhi_{33}=r \sin^{2}\theta e^{-\lambda}X_{1}+\sin\theta \cos\theta X_{2},
  \end{equation} the Ricci scalar, which from (\ref{EF}) equals $ \Phi $, is
  \begin{equation}\label{R}
  	\begin{split}
  		R=e^{-\lambda}(\lambda^{\prime\prime}-{\lambda^{\prime}}^{2}+\frac{4}{r}\lambda^{\prime}-\frac{2}{r^{2}})+\frac{2}{r^{2}},
  	\end{split}
  \end{equation}
  and the corresponding components of the Einstein tensor are:
  \begin{equation}
  	\begin{split}
  		G_{00}=\frac{1}{r^{2}}e^{-2\lambda}(r\lambda^{\prime}-1)+\frac{e^{-\lambda}}{r^{2}},
  	\end{split}
  \end{equation}
  \begin{equation}
  	\begin{split}
  		G_{11}=\frac{1}{r^{2}}(1-r\lambda^{\prime})-\frac{e^{\lambda}}{r^{2}},
  	\end{split}
  \end{equation}
  \begin{equation}
  	\begin{split}
  		G_{22}=\frac{r^{2}e^{-\lambda}}{2}(-\lambda^{\prime\prime}+{\lambda^{\prime}}^{2}-\frac{2\lambda^{\prime}}{r}),
  	\end{split}
  \end{equation}
  \begin{equation}
  	\begin{split}
  		G_{33}={\sin\theta}^2[\frac{r^{2}e^{-\lambda}}{2}(-\lambda^{\prime \prime}+{\lambda^{\prime}}^{2}-\frac{2\lambda^{\prime}}{r})]
  	\end{split}
  \end{equation} where the prime denotes $ \partial_{1} $.  \par Since $ e^{2\lambda}(\varPhi_{00}+G_{00})+\varPhi_{11}+G_{11}=0 $ from (\ref{EF}),
  \begin{equation}\label{phiG0011}
  	\begin{split}
  		\mu\frac{\lambda^{\prime}}{2}X_{1}+(X_{1}^{\prime}-\frac{\lambda^{\prime}}{2}X_{1})(1+2u_{1}u^{1})=0.
  	\end{split}
  \end{equation} 
  \par  $G_{22}+\varPhi_{22}=0 $ gives \begin{equation}\label{phiG22}
  	\begin{split}
  		-\lambda^{\prime\prime}+\lambda^{\prime2}-\frac{2}{r}\lambda^{\prime}+\frac{2e^{\lambda}}{r^{2}}(\partial_{2}X_{2}+re^{-\lambda}X_{1})(1+2u_{2}u^{2})=0
  	\end{split}
  \end{equation}
  
  and $G_{33}+\varPhi_{33}=0 $ in the interval $ 0<\theta<\pi $ yields
  \begin{equation}\label{phiG33}
  	\begin{split}
  		-\lambda^{\prime\prime}+\lambda^{\prime 2}-\frac{2}{r}\lambda^{\prime}+\frac{2e^{\lambda}}{r^{2}}(re^{-\lambda}X_{1}+X_{2}\cot\theta)=0.
  	\end{split}
  \end{equation} Subtracting (\ref{phiG22}) from (\ref{phiG33}) requires
  \begin{equation}\label{C2}
  	X_{2}\cot\theta -\partial_{2}X_{2}-2u_{2}u^{2}(\partial_{2}X_{2}+re^{-\lambda}X_{1})=0.
  \end{equation} 
  \par Using (\ref{uu1}), an equation involving the two independent components $ X_{1} $ and $ X_{2} $ can be obtained from (\ref{phiG0011}) and (\ref{C2}):
  \begin{equation}\label{C}
  	\begin{split}
  		X_2\cot\theta-\partial_{2}X_{2}+\mu(1+\frac{\lambda^{\prime}X_{1}}{\lambda^{\prime}X_{1}-2X_{1}^{\prime}})(\partial_{2}X_{2}+re^{-\lambda}X_{1})=0.
  	\end{split}
  \end{equation} \par An expression for $ X_{1}$ of the form
  \begin{equation}\label{X1}
  	X_{1}=e^{\lambda}P 
  \end{equation}
  is chosen where $ P $ is a polynomial in $ r $. Equation (\ref{C}), with $ X_{1} $ given by (\ref{X1}), then becomes 
  \begin{equation}\label{N}
  	\partial_{2}X_{2}-m(r)X_{2}\cot\theta=n(r)
  \end{equation}where $ n(r):=\frac{2\mu rPP^{\prime}}{\lambda^{\prime}P+2P^{\prime}(1-\mu)}$ and $m(r)=\frac{\lambda^{\prime}P+2P^{\prime}}{\lambda^{\prime}P+2P^{\prime}(1-\mu)}$, which has the solution 
  \begin{equation}
  	\begin{split}
  		X_{2}=\frac{n}{1-m}\sin\theta\;_{2}F_{1}(\frac{1}{2},\frac{1-m}{2};\frac{3-m}{2};\sin^{2}\theta)+c_{10}\sin^{m}(\theta).
  	\end{split}
  \end{equation}$_{2}F_{1}(\alpha,\beta;\gamma;z)  $ is the Gaussian hypergeometric function with $ \alpha=\frac{1}{2} $, $ \beta=\frac{1-m}{2} $, $ \gamma=\frac{3-m}{2} $ and $ z=\sin^{2}\theta $. Setting the arbitrary constant $ c_{10} $ to zero and using the Euler transformation $_{2}F_{1}(\alpha,\beta;\gamma;z)=(1-z)^{-\alpha}\;  _{2}F_{1}(\alpha,\gamma-\beta;\gamma;\frac{z}{z-1})   $ for $0<\theta<\frac{\pi}{4}  $ yields
  \begin{equation}\label{Fabc}
  	\begin{split}
  		X_{2}&=\frac{n}{1-m}\tan\theta\;_{2}F_{1}(\frac{1}{2},1;\frac{3-m}{2};-\tan^{2}\theta)\\
  		&=-rP\tan\theta(1+\Sigma_{n=1}^{\infty}\frac{(\alpha)_{n}(\beta)_{n}(-\tan^{2}\theta)^{n}}{(\gamma)_{n}n!})\\
  		&=-rP\tan\theta(1+\frac{1}{3-m}(-\tan^{2}\theta) +O(\tan^{4}\theta) )
  	\end{split}
  \end{equation}where $ \alpha=\frac{1}{2} $, $ \beta=1 $, $ \gamma=\frac{3-m}{2} $ and $ (\alpha)_{n}=\alpha(\alpha+1)...(\alpha+n-1)\enspace n>0 $ is the Pochhammer symbol. Using the first term of (\ref{Fabc}) as a solution for $ X_{2} $ in (\ref{phiG33}) gives the Schwarzschild solution. 
  \par  The polynomial for $P $ is assumed to be 
  \begin{equation}\label{P}
  	P=a_{0}+\frac{a_{1}}{r}+\frac{a_{2}}{r^{2}}
  \end{equation} where the $ a_{i}\;(i=0,1,2)$ are real arbitrary parameters. From (\ref{Fabc}), an expression for $X_{2}$ of the form
  \begin{equation}\label{X2B}
  	X_{2}=-\tan\theta(P+Q\theta^{2})r
  \end{equation}is sought where $Q(r)$ is the polynomial
  \begin{equation}\label{Q}
  	Q(r)=B_{2}+\frac{B_{0}}{r}+\frac{B_{1}}{r^{2}}
  \end{equation} with real parameters $B_{i}$, and $\theta>0$ is infinitesimally small so that $\tan\theta$ is represented by $\theta$ evaluated at $\theta=\theta_{0}$, the infinitesimal value of $\theta>0$. The parameters $a_{i}$ are independent of all $B_{i}$. A consistency check of (\ref{N}) with (\ref{X2B}) demands $-rP(1-m)+r\theta_{0}^{2}(mQ-3Q-P)=n$, which holds to the infinitesimal error $(-3Q-P+mQ)\theta_{0}^{2}$ because $-rP(1-m)=n$. Expressing $X_{2}$ in polynomial form to order $\theta^{2}$ gives
  \begin{equation}
  	X_{2}=-\tan\theta_{0}((a_{0}+B_{2}\theta_{0}^{2})r+a_{1}+B_{0}\theta_{0}^{2}+\frac{a_{2}+B_{1}\theta_{0}^{2}}{r})
  \end{equation}
  \par  The line element covectors are bounded, so the first term in $ X_{2} $ must vanish, which requires $B_{2}=-\frac{a_{0}}{\theta_{0}^{2}}$ . The physically relevant solution for $ X_{2} $ is then
  \begin{equation}\label{X2}
  	X_{2}=-\tan\theta(b_{0}+\frac{b_{1}}{r})\mid_{\theta_{0}}
  \end{equation} where $b_{0}:=a_{1}+B_{0}\theta^{2}_{0}$ and $b_{1}:=a_{2}+B_{1}\theta^{2}_{0}$.
  \par Equation $(\ref{phiG33})$ can now be written as 
  \begin{equation}\label{phiG33a}
  	\begin{split}
  		-\lambda^{\prime\prime}+\lambda^{\prime 2}-\frac{2}{r}\lambda^{\prime}+\frac{2e^{\lambda}}{r^{2}}(a_{0}r+a_{1}+\frac{a_{2}}{r}-b_{0}-\frac{b_{1}}{r})=0.
  	\end{split}
  \end{equation}
  By demanding 
  \begin{equation}\label{b1}
  	b_{1}=a_{2}, 
  \end{equation} the nonphysical term $ \frac{\ln r}{r} $ in the solution for $ \lambda $ can be eliminated. All polynomial expressions for $P$ and $Q$  with $\frac{1}{r^{2}}$ or higher inverse powers of $r$ generate nonphysical terms in $\lambda$, which limits the form of the polynomials to those in (\ref{P}) and (\ref{Q}). Equation (\ref{phiG33a}) then simplifies to
  \begin{equation}\label{main}
  	\begin{split}
  		-\lambda^{\prime\prime}+\lambda^{\prime 2}-\frac{2}{r}\lambda^{\prime}+\frac{2e^{\lambda}}{r^{2}}(a_{0}r-b)=0
  	\end{split}
  \end{equation}where
  \begin{equation}
  	-b=a_{1}-b_{0},
  \end{equation} which has the exact solution
  \begin{equation}\label{wow}
  	\begin{split}
  		\lambda=-\ln(-a_{0}r+2b\ln r+\frac{c_{1}}{r}+c_{2}), \enspace 0<r<\infty.
  	\end{split}
  \end{equation}The arbitrary parameters $c_{1}$ and $ c_{2} $ are those in the Schwarzschild solution: $c_{1}=-\frac{2GM}{c^{2}}$ and $c_{2}=1$. Equation (\ref{wow}) represents the extended Schwarzschild solution. All solutions to the Einstein equation in MGR are bounded, so (\ref{wow}) demands $ r $ to be positive and finite; $ r>0 $ can be as large as necessary to describe any physically reasonable Universe or part thereof, but it cannot extend to infinity. It is clear from (\ref{Newt}) that the parameters $a_{0}$ and $\mid b\mid$ are related to dark energy and dark matter, respectively. 
 
 \subsection{The gravitational energy density}
\par The relationship between the parameters can be determined from the local energy density of the gravitational field, which follows from (\ref{T}):
\begin{equation}\label{W}
	E=-\frac{c^{4}}{8\pi G}\varPhi_{00}.
\end{equation}  From $\varPhi_{00}=\frac{\mu}{2}X_{1}e^{-2\lambda}\lambda^{\prime}$, (\ref{X1}), (\ref{P}), and (\ref{wow}), 
\begin{equation}
\varPhi_{00}=\frac{\mu}{2}(a_{0}-\frac{2b}{r}+\frac{c_{1}}{r^{2}})(a_{0}+\frac{a_{1}}{r}+\frac{a_{2}}{r^{2}}),
\end{equation}which must yield the Newtonian, dark energy, and dark matter local gravitational energy densities. By setting the coefficients of the $ \frac{1}{r}$ and $\frac{1}{r^{3}} $ terms to zero, we obtain 
\begin{equation}\label{a1}
2b=a_{1} 	
\end{equation}
and $ 2ba_{2}=c_{1}a_{1} $, respectively, from which
\begin{equation}\label{a2}
a_{2}=c_{1}. 
\end{equation} The $\frac{1}{r^{4}}$ term exhibits the Newtonian behavior, so a contribution from dark matter must come from the $\frac{1}{r^{2}}$ term, which vanishes if $a_{1}^{2}=2a_{0}\mid c_{1}\mid$. Thus, to prevent that, we set 
\begin{equation}\label{a0c1}
2a_{0}\gamma\mid c_{1}\mid=a_{1}^{2} 
\end{equation}where $\gamma>1$ is a parameter that relates dark matter $M_{DM}$ to ordinary (baryonic) matter $M$ by the ratio 
\begin{equation}\label{gamma}
M_{DM}=\gamma M.
\end{equation} The mean ratio of dark matter to ordinary matter in the Universe is approximately five, which is consistent with $\gamma>1$; the amount of dark matter is always greater than that of ordinary matter. This requires dark matter to exist within the space occupied by ordinary matter, which demands that dark and ordinary matter do not interact, as advocated by the $\Lambda$CDM standard model. $\gamma$ varies widely depending on the dark matter content of the cosmological entities; ultrafaint dwarf spheroidals are heavily dark matter-dominated \cite{Lelli} and $\gamma$ varies between 10-1000, whereas Section 4 shows that $\gamma$ slightly above one explains the wide binaries anomaly.  Since the contained mass at a particular radius $r$ in a galaxy or galactic cluster describes their rotation curves, (\ref{gamma}) is more generally stated by $M(r)_{DM}=\gamma(r) M(r) $ where  $\gamma(r)$ locally defines the ratio of the dark matter to baryonic matter.
\par  From (\ref{a1}) and (\ref{a0c1}),
\begin{equation}\label{b}
b=\pm\sqrt{\frac{a_{0}\mid c_{1}\mid\gamma}{2}}.
\end{equation} Hence,
\begin{equation}\label{gen}
\varPhi_{00}=\mu(\frac{1}{2}a_{0}^{2}+\frac{4b^{2}(\frac{1}{\gamma}-1)}{r^{2}}+\frac{c_{1}^{2}}{2r^{4}})
\end{equation} and it follows that the local energy density of the static spherically symmetric gravitational field is
\begin{equation}\label{W0}
E=-\frac{\mu c^{4}}{16\pi G}a_0^{2}+\frac{\mu c^{2}a_{0}M(\gamma-1)}{2\pi r^{2}}-\frac{\mu GM^{2}}{4\pi r^{4}}\;\;\gamma>1.
\end{equation}The first term is the local gravitational energy density of dark energy, and the second is that from dark matter. The third term is the Newtonian gravitational energy density \cite{LL} if $\mu=\frac{1}{2}$.

\subsubsection{The parameter \texorpdfstring{$a_{0}$}{a0}}
Since $c_{1}=-\frac{2GM}{c^{2}}$, $a_{0}$ is the only independent parameter in the extended Schwarzschild solution because $b$ is determined from it. It is directly related to dark energy in the spherically symmetric solution of the Einstein equation. The line element covector field $X_{\beta}$ exists everywhere, so $a_{0}$ in the radial component $X_{1}$ given by (\ref{X1}) and (\ref{P}) must have a value for each $X_{1}$ in spacetime. The parameter $a_{0}$ is generally determined when the dark energy acceleration equals the Newtonian acceleration by (\ref{NDE}), or when the dark matter acceleration equals the dark energy acceleration in (\ref{inrot}). However, there are situations when there are no such equilibriums and a scaling relation for $a_{0}$ must be invoked, such as $a_{0}=\zeta l_{p}\frac{M}{M_{\odot}}$  as developed in \cite{ND} where $\zeta=3.55\times10^{-6}$ is a parameter of dimension $L^{-2}$ determined from the dark matter halo of galaxy NGC3198, $l_{p}$ is the Planck length, and $M_{\odot}$ is the ordinary mass of the Sun. This scaling relation works well in the planetary part of the Solar system and in some galactic clusters such as SDSS J0900+2234, but must not be used when $a_{0}$ can be calculated from the stated equilibrium conditions.

\section{The radial acceleration and Tully-Fisher relations, and MOND}
The radial gravitational force on a test mass $m$ can now be calculated from (\ref{wow}). Using the weak-field relationship of the Newtonian potential $ \phi $ to $ g_{00} $,
\begin{equation}\label{phi}
	\phi=\frac{c^{2}}{2}(e^{-\lambda}-1),	
\end{equation} the modified radial Newtonian force on $m$ is
\begin{equation}\label{Newt}
	F_{r}=-\frac{GMm}{r^{2}}-\frac{\mid b\mid mc^{2}}{r}+\frac{a_{0}mc^{2}}{2}
\end{equation} where $M$ represents the total baryonic (ordinary) mass of the cosmological structure.\par The correction terms to the Newtonian force come from the non-zero components of the line element field in the energy-momentum tensor $ \varPhi_{\alpha\beta}$. The second term is gravitationally attractive and represents the correction from invisible mass. That term provides the additional gravitational attraction that is missing in GR. Since the parameter $ b $ has two signs, its absolute value is invoked to ensure that term is negative. \par The third term is positive and repulsive with $ a_{0}>0 $ as defined. This describes the repulsive dark energy force in the present accelerating epoch. However, a decelerating epoch previous to the current accelerating epoch has been observed by Riess et al. \cite{Riess}. They used the Hubble telescope to provide the first conclusive evidence for cosmic deceleration that preceded the current epoch of cosmic acceleration. Equation (\ref{Newt}) describes the competition between the dark matter and dark energy forces relative to the Newtonian force for a decelerating or accelerating Universe.
\par The weak field force equation (\ref{Newt}) is generally used under equilibrium conditions between the dark energy and Newtonian forces, or the equilibrium between the dark energy and dark matter forces.  Both situations show $a_{0}\propto \frac{GM}{r^{2}} $, and $b\propto\frac{\sqrt{GM}}{r}$. The weak field approximation remains small as $r$ increases for any cosmic application. The metric parameter
$\lambda$ is bounded for any finite $r>0$. Clearly, the log term in $\lambda$ does not appear in the gravitational energy density, the modified force equation (\ref{Newt}), or the rotation velocity (\ref{V}). Thus, no issues from it arise from subsequent use of these expressions.
\par Assuming a circular orbit of a star in a galaxy with a symmetric gravitational potential, it follows that 
\begin{equation}\label{parphi}
	r\frac{\partial \phi}{\partial r}=v^{2},
\end{equation}
 and the rotational velocity of the star satisfies  
\begin{equation}\label{V}
	v^{2}(r)=\frac{GM(r)}{r}+\mid b\mid c^{2}-\frac{a_{0}c^{2}}{2}r
\end{equation} where $ M(r) $ is the total ordinary mass of the galaxy interior to a fixed radius $ r $. Equation (\ref{V}) also applies to a galactic cluster; it demands an upper limit to r describing a large but finite cosmological structure.  

\subsection{The radial acceleration relation}
The radial acceleration relation describes the local link between baryons and dynamics in galaxies; it considers each individual point along the flat part of the rotation curve and the corresponding enclosed baryonic mass \cite{Lelli}. Although the flat part of the rotation curve generally pertains to its outer part, the RAR can be extended to the inner part of the rotation curve as described below.
\par The RAR is established for rotationally supported disk galaxies with extended stellar
and gas distributions. The disk component of a galaxy contains most of the stars, gas, and dust in a near-circular coplanar orbit around the galactic center, so a disk galaxy is generally composed of the disk(d), a bulge(b), a dark matter halo(DM), and a black hole(BH). Since dark energy (DE) exists everywhere, it must also be included as a gravitational component of the galaxy. The gravitational potential of each component of the galaxy can be expressed in the axisymmetric coordinates $(R,z)$ where the origin is at the galactic center, the $z$ axis is aligned with the galaxy's symmetry axis, and the radius $R$ is restricted to the plane of the disc. The superposition of the components of the weak-field gravitational potential of the galaxy demands 
\begin{equation}
	\phi(R,z)=\Sigma_{i}\varphi_{i}(R,z),
\end{equation} where $\phi_{i}(R,z)$ is the gravitational potential of the i-th component. The squared rotational velocity $v^{2}(R,z)$ of the disc galaxy follows from (\ref{parphi}):
\begin{equation}
	v^{2}(R,z)=R\Sigma_{i}\partial_{R}\phi_{i}(R,z)=R\partial_{R}\phi(R,z),
\end{equation}which is equivalent to
\begin{equation}\label{sumV}
	v^{2}(R,z)=v_{OM}^{2}+v_{DM}^{2}-v_{DE}^{2}
\end{equation} where $v_{OM}^{2}=v_{d}^{2}+v_{b}^{2}+v_{BH}^{2}$ is the weak-field gravitational effect of ordinary matter from the disk, bulge, and black hole, respectively. At a location $r$ relative to the center of the disk galaxy, $v^{2}(r)$ can be expressed in axisymmetric coordinates by the relation $r=\sqrt{R^{2}+z^{2}}$, so that the rotational velocity (\ref{sumV}) is equivalently described by (\ref{V}).

\par It is now shown how the dynamics between dark energy and ordinary matter in the outer part of the rotation curve, and between dark energy and dark matter in the inner part of the rotation curve of a galactic structure, determine the RAR and its physical origin.
  \subsubsection*{The outer rotation curve}
When the dark energy acceleration equals the Newtonian acceleration, 
\begin{equation}\label{NDE}
	\frac{GM(r)}{r^{2}}=\frac{a_{0}c^{2}}{2}.
\end{equation} It is clear from (\ref{V}) that the rotation curve is horizontally flat. This occurs at a critical radius $r_{c}$ defined by
\begin{equation}\label{rc}
	r_{c}:=\sqrt{\frac{2GM(r)}{a_{0}c^{2}}}.
\end{equation}When $r\geq r_{c}$, the outer part of the rotation curve is flat, but in general not horizontally flat as shown by (\ref{v2}). This means there is not an exact equilibrium in (\ref{NDE}) unless the deviation from equilibrium is corrected by $\alpha$ in
\begin{equation}\label{a0out}
	a_{0}=\frac{GM(r)\alpha}{c^{2}r^{2}}
\end{equation}where $\alpha>0$ is a parameter to be determined for each galactic structure at each point along the outer rotation curve. Then from (\ref{V}), the outer rotation curve has the observed acceleration 
\begin{equation}\label{aout}
	\begin{split}
	\frac{v^{2}}{r}&=\frac{GM(r)}{r^{2}}+\frac{GM(r)\sqrt{\alpha\gamma}}{r^{2}}-\frac{GM(r)\alpha}{2r^{2}}\\
	&=\frac{GM(r)}{r^{2}}(1+\sqrt{\alpha\gamma}-\frac{\alpha}{2}),
		\end{split}
\end{equation}which consists of the Newtonian, dark matter, and dark energy accelerations, respectively. This equation replaces the empirical relation in (\ref{RAR}). The physical origin of the RAR in the outer part of the rotation curve is determined by the equilibrium between the dark energy and baryonic accelerations. The observable acceleration depends directly on the baryonic acceleration along the radial direction commensurate with the definition of the RAR.
\par The product $\alpha\gamma$ is derived from (\ref{a0out}) and (\ref{TF4}):  $\alpha\gamma=\frac{v^{4}r^{2}}{G^{2}M^{2}}$. The SPARC data \cite{SPARC} was used to obtain $\alpha\gamma$ using $R_{p}$ for the radius of peak baryonic rotational velocity $V_{b}$, and $M_{b}$ for the total baryonic mass. Since $\alpha$ is the deviation from an exact equilibrium between dark energy and ordinary matter (exact when $\alpha=2$), and $\gamma$ is expected to be in the range $1<\gamma<1000$, there were 14 galaxies in the database that had $\alpha\gamma$ values greater than 2500, which were considered too large for the stated equilibrium to exist. The results are presented in Table 2: Calculation of $\alpha\gamma$ for 163 SPARC disc galaxies.
\par The value for $\gamma$ can be determined from a complete modeling of the extended rotation curves resulting from dark matter. Galaxy NGC3198 \cite{Karukes}, NGC2403 \cite{Barker}, and NGC5055 \cite{Bat,Xiao} have been modeled and the results for $\alpha$ and $\gamma$ are presented in Table 1.

\begin{table}[h!]
	\centering
	\caption{Calculation of $\alpha$ for specific galaxies}
	\begin{tabular}{lccccc}
		\hline
		Galaxy & $M_{b}\times10^{10}M_{\odot}$ & $M_{DM}\times10^{10}M_{\odot}$ & $\gamma=\frac{M_{DM}}{M_{b}}$ & $\alpha\gamma$ & $\alpha=\frac{\alpha\gamma}{\gamma}$ \\
		\hline
		NGC3198  & 3.398 & 85.6 & 25.19 & 52.43 &  2.08 \\
		\hline NGC2403  & 0.881 & 8.619 & 9.78  & 36.56 & 3.74  \\
		\hline NGC5055 & 7.539  & 142.46  & 18.90  & 60.27 & 3.19 
		 \\ 
	\end{tabular}
	\label{tab:SPARC_MGR_Final}
\end{table}
Galaxy NGC3198 has an alpha value of 2.08, which indicates its out rotation curve is nearly horizontally flat. The outer rotation curve of galaxy  NGC2403 is flat and slightly rising; that of NGC5055 is flat and slightly falling.

\subsubsection*{The inner rotation curve}
 The inner part of the rotation curve $r<r_{c}$, is typically baryon dominated and the observed acceleration is mainly Newtonian. However, dark energy and dark matter exist at every point in spacetime. When the dark matter acceleration equals the dark energy acceleration in a particular galactic structure, 
\begin{equation}\label{ab}
	\frac{\mid b\mid c^{2}}{r} =\frac{a_{0}c^{2}}{2}
\end{equation}and the observed acceleration from (\ref{V}) depends entirely on the baryonic matter commensurate with the RAR:
\begin{equation}\label{inrot}
	\frac{v^{2}}{r}=\frac{GM(r)}{r^{2}}.
\end{equation}From ((\ref{ab})),
\begin{equation}\label{a0in}
	a_{0}=\frac{4GM(r)\gamma}{c^{2}r^{2}}=\frac{4GM_{DM}(r)}{c^{2}r^{2}}
\end{equation} guarantees the rotation curve behaves Newtonian in the presence of dark matter and dark energy. 
\par Local gravitational energy-momentum, which manifests itself as dark energy and dark matter, exists at every point in space and time. The equilibrium between the dark energy and baryonic accelerations in the outer part of the galactic rotation curve, and the equilibrium of the dark energy and dark matter accelerations in the inner part of the rotation curve, is the physical origin of the RAR. The observed circular acceleration $\frac{v^{2}}{r} $, which must involve visible baryonic matter, and both the dark matter and dark energy accelerations, depends directly on the baryonic acceleration commensurate with the definition \cite{McG,Lelli} of the RAR. It holds for all galactic structures of any type and for all radii in the rotation curve; it is not limited to galaxies. Any deviation from the RAR evidences a state of non-equilibrium between the dark energy and Newtonian acceleration, or between the dark sector accelerations themselves.
\par Since the spherically symmetric solution of the Einstein equation is static because of that symmetry, the RAR should hold at all times. It was formulated mainly from late-time galaxies, recently observed to hold at intermediate redshifts \cite{Cio}, and it is shown \cite{MistMcG} that there is no significant difference between the RAR for early-time and late-time galaxies.
\subsection{The baryonic Tully-Fisher relation}
Contrary to the RAR, the
BTFR considers a single value of the asymptotically flat rotation velocity and total baryonic mass $M$ for each galaxy. From (\ref{aout}) at $r_{f}$ where the outer rotation curve is asymptotically flat and $M(r_{f})=M$, 
\begin{equation}\label{vf}
	v^{2}=Mf
\end{equation} where
\begin{equation}\label{f}
	f=\frac{G}{r_{f}}(1+\sqrt{\alpha\gamma}-\frac{\alpha}{2})
\end{equation}
From (\ref{vf}), we can write 
\begin{equation}\label{xf}
	M=Av^{2+x}
\end{equation}where $f^{-1}=Av^{x}$ and $A$ is a parameter that maintains the dimensionality of (\ref{xf}). This is the BTFR relation with $x=\frac{d\log M}{d\log v}-2$. 
\par From (\ref{V}), we can define
\begin{equation}\label{Delta}
	\Delta(r):=-\frac{GM(r)}{r^{2}}+\frac{a_{0}c^{2}}{2}	
\end{equation} where $\Delta(r)$ is the difference between the Newtonian and the dark energy accelerations, and write
\begin{equation}\label{v2}
	v^{2}=\mid\ b\mid c^{2}-r\Delta(r),
\end{equation}which holds for all finite $r>0$. If $r$ is large, (\ref{v2}) describes the observed linearly falling or rising rotation curves of most galaxies. It follows that 
\begin{equation}\label{V4}
	v^{4}+2\mid b\mid c^{2}r\Delta(r)-r^{2}\Delta(r)^{2}=b^{2}c^{4}.
\end{equation}
If $\Delta<0$, the rotation curve linearly rises for large $r$. The left hand side of (\ref{V4}) would require an exponent less than four to be equivalently expressed as a power of $v$. If $\Delta(r)>0$ and $r$ is very large, the dark energy acceleration in (\ref{v2}) dominates and we have a linearly falling rotation curve; the equivalent power of $v$ would be require an exponent less than four. However, if $2\mid b\mid c^{2}>r\Delta(r) $ and $\Delta(r)>0$, $\frac{1}{c}\sqrt{\frac{2GM(r)}{a_{0}}}<r<\frac{2}{c}\sqrt{\frac{\gamma GM(r)}{a_{0}}}(1\pm\sqrt{1+\frac{2}{\gamma}})$, which implies a power of $v$ with a value $>4$. Thus, the BTFR exponent $2+x$ is generally less than four, but can be greater than four when dark energy is dominant in a particular region of the rotation curve. The physical origin of the Tully Fisher relation follows directly from the existence of ordinary matter, dark energy, and dark matter in the rotation curve (\ref{V}) that is obtained from the dynamics of the modified Newtonian force equation (\ref{Newt}).

\subsection{The power four Tully-Fisher relation}
 When the magnitude of the Newtonian acceleration equals that of the dark energy acceleration, (\ref{NDE}) demands $\Delta(r)=0$,  which generates the horizontally flat rotation curve
 \begin{equation}\label{vhflat}
 v^{2}=\mid b \mid c^{2}
 \end{equation} representing a pure dark matter halo with the parameter $\mid b\mid$. The constant rotation speed requires an enclosed mass that increases linearly with $r$ for $r_{c}\leq r$, which follows directly from (\ref{b}) and (\ref{NDE}):
 \begin{equation}
 	M(r)=\frac{\mid b\mid c^{2}r}{\sqrt{2\gamma}G}\;\;\text{or}\;\;	M_{DM}(r)=\frac{\mid b\mid c^{2}\sqrt{\gamma}r}{\sqrt{2G}}.
 \end{equation}  The linear behavior does not not hold indefinitely; it must end at the maximum value of $r=r_{d}$ defined in (\ref{rd}).
  \par From (\ref{vhflat}), we obtain the power-four TF relation
 \begin{equation}\label{TF4}
 	v^{4}=c^{2}a_{0}\gamma GM(r)=c^{2}a_{0}GM(r)_{DM},
 \end{equation}
 
  \subsubsection{MOND}
 Modified Newtonian Dynamics (MOND) defines an acceleration scale below which Newtonian dynamics is replaced by a gravitational acceleration that is the square root of the Newtonian acceleration $a_{N}$: $a=\sqrt{A_{0}a_{N}}=\sqrt{A_{0}\frac{GM}{r^{2}}}$. The acceleration constant $A_{0}$, although unknown from  first principles, must be the same for all galaxies \cite{Gent} and is generally set at $A_{0}=1.2\times10^{-10} $. MOND has a $v^{4}$ structure that follows directly from (\ref{TF4}):
 \begin{equation}\label{MOND}
 v^{4}=A_{0}GM(r)
 \end{equation}  where the MOND acceleration $A_{0}$ is defined in terms of the MGR parameters as  
 \begin{equation}\label{A0}
 	A_{0}:=c^{2}a_{0}\gamma,
 \end{equation} which proves MGR subsumes MOND. From (\ref{a0out}), $A_{0}=\alpha\gamma\frac{GM(r)}{r^{2}}$, so $A_{0}$ is constant if and only if those parameters change accordingly with the Newtonian acceleration for each galaxy. MGR reproduces the universal acceleration of MOND when the Newtonian acceleration is small. However, if $\frac{GM(r)}{r^{2}}$ is not in the small range of $A_{0}$, the Newtonian acceleration dominates, which demands equilibrium between the dark sector accelerations; (\ref{a0in}) holds and $ A_{0}=\frac{4GM(r)\gamma^{2}}{r^{2}}$. 
 \par Thus, MGR explains MOND relative to the critical radius (\ref{rc}); if $r\geq r_{c}$, there is an equilibrium between the Newtonian and dark energy forces and $A_{0}=\alpha\gamma\frac{GM(r)}{r^{2}}$ is small and constant because the parameters $\alpha$ and $\gamma$ are galaxy dependent and compensate changes in the Newtonian acceleration. When $r<r_{c}$, the dark sector forces are in equilibrium and the Newtonian gravity precisely holds. MGR is a covariant theory that subsumes the MOND paradigm at a fundamental level.
 
 \par Although $A_{0}$ is considered to be a fundamental constant, horizontally flat galaxies with the $v^{4}$ behavior in (\ref{MOND}) cannot be described with the constant acceleration of $A_{0}=1.2\times10^{-10}$. For example, the horizontally flat power four  galaxy NGC 3198 with $M=4.4\times10^{10}M_{\odot} $, (\ref{MOND}) generates a rotation velocity of 162.7km/s, but the observed rotation velocity \cite{Karukes} for $r>17\text{kpc}$ is 150km/s; similarly, NGC 6503 with $M=1.9\times10^{40}kg$ has a MOND velocity of 111.1km/s compared to the observed value of 116km/s in Fig. 11 \cite{Greis} for $r>3.5\text{kpc}$.  These counter-examples prove that MOND does not have a pure power four circular velocity behavior with the constant acceleration $A_{0}$. Nevertheless, MOND lessens the strict power four rotational velocity dependence and maintains the constancy of $A_{0}$ from a $v^{y}$ curve with $y=3.94$ for NGC 3198, and $y=4.04$ for NGC 6503. 
\par A fundamental difference between MGR and MOND is that MOND fails to account for the gravitational lensing of galaxy clusters. In contrast, MGR successfully describes \cite{ND} the lensing of galactic cluster SDSS J0900+2234 with an Einstein ring. Moreover, MOND cannot describe the observed \cite{Gall} Newtonian acceleration between galactic clusters that MGR explains in Section 5.
 
  \subsection{Extended flat rotation curves}
Flat circular rotation curves have been measured that show no clear indication of a decline \cite{Mist} out to 1Mpc. This behavior
persists for both early and late-type galaxies, and can be explained by splitting the dark matter term $\mid b\mid c^{2}=(\mid \bar{b}\mid+\mid \tilde{b}\mid)c^{2}$ with $\sqrt{\gamma}=\sqrt{\bar{\gamma}}+\sqrt{\tilde{\gamma}}$ so that
 (\ref{V}) becomes 
\begin{equation}\label{bave}
	v^{2}=\mid \bar{b}\mid c^{2}
\end{equation} if
\begin{equation}\label{tildeb}
	\frac{GM}{r}+\mid \tilde{b}\mid c^{2}=\frac{a_{0}c^{2}}{2}r.
\end{equation} This quadratic equation in $r$ for a fixed enclosed baryonic mass has the solution
\begin{equation}\label{rd}
	\begin{split}
		r_{d}&=\frac{1}{c}\sqrt{\frac{\tilde{\gamma}GM}{a_{0}}}(1\pm\sqrt{1+\frac{2}{\tilde{\gamma}}})\\
		&=r_{c}\sqrt{\frac{\tilde{\gamma}}{2}}(1\pm\sqrt{1+\frac{2}{\tilde{\gamma}}})
	\end{split}
\end{equation}using (\ref{rc}). The length of the rotation curve beyond the fixed point $r_{c}$ where it first becomes horizontally flat is determined by $\tilde{\gamma}$, which depends on $\gamma$, the ratio of the dark matter to ordinary matter in (\ref{gamma}). The length of the extended flat rotation curve depends entirely on the amount of dark matter in the galaxy or galactic cluster. The dark matter halo of any galaxy cannot extend indefinitely because dark matter is attractive and proportional to ordinary matter by $M_{DM}=\gamma M$; $\gamma$ is finite and galaxy dependent.

\section{Wide binaries}
Studies of the Gaia data in the low acceleration regime of solar mass wide binary stars evidence a clear gravitational anomaly above separations of a few kau  of the type $ G\rightarrow\beta G$ with $\beta=1.43\pm 0.06$ consistent with MOND AQUAL expectations of
$\beta\approx1.4$\cite{Chae,Hern2} and $\beta=1.5\pm 0.2$\cite{Hern1}, ruling out a low acceleration Newtonian behavior. This deviation from the Newtonian gravitational constant is thought to be from the effect of an external gravitational field. However, dark matter and dark energy exist at every point in spacetime; they can equally explain the external field effect in wide binaries.
\par The Newtonian acceleration is small in wide binary systems and it is possible for it to equal the dark energy acceleration. Then (\ref{NDE}) generates
\begin{equation}
	a_{0}=\frac{2GM(r)}{r^{2}c^{2}}.
\end{equation}A wide binary system of ordinary mass $M=1.5M_{\odot}$ at a separation of $r_{WB}=20\text{kau}$ generates the Newtonian, dark matter, and dark energy accelerations $a_{N}=2.23\times10^{-11}ms^{-2}$, $a_{DM}=2.23\sqrt{2\gamma}\times10^{-11}ms^{-2}$, $a_{DE}=2.23\times10^{-11}ms^{-2}$, respectively. In the low gravitational regime, $\gamma$ is expected to be slightly above one; if $\gamma=1.01$, $\sqrt{2\gamma}=1.43$ and if $\gamma=1.1$, $\sqrt{2\gamma}=1.56$, which is commensurate with \cite{Chae} and \cite{Hern1,Hern2}, respectively. \par MGR explains the wide binaries anomaly without the effect of the gravitational field of baryonic masses external to the wide binary system, and without modifying the gravitational constant. The strong equivalence principle is not violated in MGR. 
\section{Large scale gravity between galaxy clusters}
The gravitational acceleration between galactic clusters was measured using the Atacama Cosmology Telescope. On scales from 30 to 230 megaparsecs, the gravitational acceleration between pairs of galactic clusters at separation $r$ was found to be proportional to $\frac{1}{r^{n}}$ with $n=2.1\pm 0.3$\cite{Gall}. This is consistent with Newtonian gravity in accordance with the standard $\Lambda$CDM model, where it holds incredibly well terrestrially and within individual galaxies. This study shows dark matter in the $\Lambda$CDM model is required to explain the rotations of galaxies and the movements of galaxies within clusters; it severely constrains modified theories of gravity that dismiss dark matter. It specifically shows that MOND cannot account for large scale gravitation because it varies as $\frac{1}{r}$, not $\frac{1}{r^{2}}$. 
\par The study was based on the cosmological parameters that can be obtained from the pairwise velocity between any two galaxies or galactic clusters (Clusters) in an expanding background described by the FLRW metric. The proper separation of the Clusters located at coordinates $\vec{x_{1}}$ and $\vec{x_{2}}$ comoving in an expanding background is given by
\begin{equation}
	\vec{r_{12}}=a(t)\vec{x_{12}(t)} 
\end{equation}where $\vec{x_{12}}=\vec{x_{1}}-\vec{x_{2}}$. Using the peculiar velocity $\vec{v_{i}}=a\dot{\vec{x_{i}}}\;\;i=1,2$, the pairwise velocity is $\vec{v_{12}}=\vec{v_{1}}-\vec{v_{2}}+\frac{\dot{a}}{a}\vec{r_{12}}$.
\par Rather than using the pairwise velocity to investigate the large scale nature of gravity, it is interesting to explore the possibility that dark energy described by $\varPhi_{00}$ may have an inherent structure that leads to the Newtonian gravitational behavior. Consider two Clusters located at fixed coordinates $\vec{x_{1}}$ and $\vec{x_{2}}$ comoving in an expanding open Universe described by the FLRW metric with $ \kappa=1$. Then $\vec{v_{i}}=0$ and the relative radial acceleration between the Clusters is $\frac{\ddot{r_{12}}}{r_{12}}=\frac{\ddot{a}}{a} $. From (\ref{F1}), the relative radial gravitational acceleration of the Clusters in an expanding Universe filled with dark energy and dark matter is
\begin{equation}\label{r12tot}
	\frac{\ddot{r_{12}}}{r_{12}}=-\frac{4\pi G}{3}(\varrho_{tot}+\frac{3p_{tot}}{c^{2}})+\frac{c^{2}}{6}(\varPhi_{00}+3\Phi^{1}_{1})
\end{equation}where $\varrho_{tot}=\varrho_{1}+\varrho_{2}$ and $p_{tot}=p_{1}+p_{2}$ is the total baryonic mass density and pressure, respectively. 
\par The dark energy of the Clusters in an open Universe 
\begin{equation}
	\varPhi_{00}=\Lambda+\frac{8\pi G\varrho_{tot}}{c^{2}}-\frac{3x^{2}_{12}}{r_{12}^{2}}-\frac{6}{c^{2}}\int \frac{\dot{H}}{a}da
\end{equation}follows from (\ref{1Phi00}), so the relative radial gravitational acceleration of the Clusters is
\begin{equation}
	\frac{\ddot{r_{12}}}{r_{12}}=-\frac{c^{2}x^{2}_{12}}{2r_{12}^{2}}+\frac{\Lambda c^{2}}{6}-\frac{4\pi Gp_{tot}}{c^{2}}-\int\frac{\dot{H}}{a}da+\frac{c^{2}\varPhi^{1}_{1}}{2}
\end{equation}where $r_{12}=\sqrt{r^{2}_{1}+r^{2}_{2}-2r_{1}r_{2}\cos\theta}$ with $\theta$ the angle between the radial vectors $\vec{r_{1}}$ and $\vec{r_{2}}$. 
\par Thus indeed, dark energy inherently contains the $\frac{1}{r^{2}_{12}}$ structure that is responsible for the Newtonian gravitational behavior. This result holds for all scales because there were no restrictions on the scale of the expanding Clusters. Moreover, the fixed comoving coordinates can be applied to any two objects in the expanding Universe. Since dark energy exists everywhere in spacetime, the Newtonian behavior of gravity fundamentally depends on dark energy.

\section{Discussion and conclusions}
 The tight correlation between the distribution of baryonic matter in galaxies and their observed dynamics, even in systems where dark matter dominates, suggests a universal acceleration scale exists in galaxies. This universal acceleration, known as the RAR, appears to challenge the $\Lambda$CDM standard cosmological model and favor MOND. However, MOND also deviates from the RAR, so these deviations are not fully explained by the $\Lambda$CDM model or MOND. 
 \par MGR explains gravity and the dark sector geometrically and describes dark matter as a spin-1 particle, which eliminates the quandary between the standard model and MOND. MGR is the only metric-compatible theory of gravity that describes local gravitational energy-momentum by the connection-independent symmetric tensor $\varPhi_{\alpha\beta}$, which completes the total energy-momentum tensor of the Einstein equation. Dark energy is explained by the energy-related component $\varPhi_{00}$, and dark matter by the stress components $\varPhi_{ij}$. Dark matter is extremely difficult to observe because it is fundamentally related to the Lorentzian metric of spacetime itself. 
 \par The modified Newtonian force equation in MGR contains the dark energy and dark matter forces of the dark sector in addition to the Newtonian force of GR for weak gravitational fields. The physical origin of the RAR follows directly from the equilibrium between the dark energy and Newtonian forces in the outer part of the rotation curve of galaxies or galaxy clusters, and that between the dark sector forces in the inner part of the rotation curve. The observed acceleration in the RAR depends directly on the baryonic acceleration of these structures. It was specifically shown that
\begin{itemize}
	\item  MGR is a covariant theory that subsumes the MOND paradigm at a fundamental level.  It is more general than MOND because both $a_{0}$ and $\gamma$ are not fixed parameters in MGR.
	\item The constant parameter in MOND cannot satisfy a pure $v^{4}$ power law for all galaxies. Rather, MOND approximates the exponent while keeping $A_{0}$ fixed.
	\item The Tully-Fisher relation holds for the horizontally flat and linearly rising or falling outer rotation curves. The exponent of the rotational velocity is generally less than four but can be greater than four under certain conditions;
	\item  The length of the extended flat rotation curve depends entirely on the amount of dark matter in the galaxy or galactic cluster. It cannot extend indefinitely.
	\item MGR explains the wide binaries anomaly without the effect of the gravitational field of baryonic masses external to the wide binary system, and without modifying the gravitational constant $G$. 
	\item The Newtonian radial gravitational acceleration holds for all scales in MGR, including the enormously large scale of galactic clusters. The structure of dark energy described by $\varPhi_{00}$ explains the nature of Newtonian gravity in the Universe.
\end{itemize}

\section*{Acknowledgments} I would like to thank Dr. Karsten M{\"u}ller for the many interesting discussions about dark matter and Modified General Relativity.
\appendix
\section {Orthogonal Decomposition Theorem}
\begin{theorem}
		A non-divergenceless (0,2) symmetric tensor $ \gamma_{\alpha\beta} $ in the symmetric cotangent bundle $ \Gamma^{\infty}_{c} (S^{2}T^{\ast}\mathcal{M}) $ with smooth sections of compact support on an n-dimensional noncompact paracompact boundaryless time-oriented Lorentzian manifold $ \mathcal{M} $ with a Levi-Civita connection can be orthogonally decomposed as $
	\gamma_{\alpha\beta}= v_{\alpha\beta}+ \varPhi_{\alpha\beta} $ where $v_{\alpha\beta}  $ represents a linear sum of symmetric divergenceless (0,2) tensors and $\varPhi_{\alpha\beta}=\frac{1}{2}\mathcal{L}_{X}g_{\alpha\beta}+\mathcal{L}_{X}(u_{\alpha}u_{\beta})$. The timelike unit vector $\bm{u}  $ is collinear with one of the pair of regular vectors in the line element field $ (\bm{X},-\bm{X}) $ and $\bm X$ is not a Killing vector.
\end{theorem}
\begin{proof}
	Let the Lorentzian manifold with metric $ (\mathcal{M},g_{\alpha\beta}) $ be noncompact paracompact, time-oriented, and boundaryless. A smooth regular line element field $(\bm{X},\bm{-X)}$ exists such that $\bm{X}$ is not a Killing vector field, as does a timelike unit vector $ \bm{u} $ collinear with one of the pair of line element vectors. Let $ \mathcal{M} $ be endowed with a smooth Riemannian metric $ g^{+}_{\alpha\beta} $. The smooth Lorentzian metric $ g_{\alpha\beta} $ is constructed from $ g^{+}_{\alpha\beta} $ and the unit covectors $ u_{\alpha}$ and $ u_{\beta}$ : $g_{\alpha\beta}=g^{+}_{\alpha\beta}-2u_{\alpha}u_{\beta} $. Let $\gamma $ and $ v $ belong to $\Gamma^{\infty}_{c} (S^{2}T^{\ast}\mathcal{M}) $, the cotangent bundle of symmetric $(0,2)$ tensors on $ \mathcal{M} $ with smooth sections of compact support, and let $K$ be a compact subset of $\Gamma^{\infty}_{c} (S^{2}T^{\ast}\mathcal{M})$  containing the support of $\delta\gamma $ such that $\delta\gamma$ is orthogonal to every $1$-form $\omega$ with the property that $\delta\delta^*\omega$ restricted to $K$ vanishes. A non-divergenceless $ (0,2) $ symmetric tensor $ \gamma_{\alpha\beta} $ can be orthogonally and uniquely decomposed with respect to $ g^{+}_{\alpha\beta} $ according to $ \gamma_{\alpha\beta}=v_{\alpha\beta}+\frac{1}{2}{\mathcal L}_Xg^{+}_{\alpha\beta}$ where $v_{\alpha\beta}  $ represents a linear sum of symmetric divergenceless (0,2) tensors and $ {\nabla^{+}}^{\alpha}v_{\alpha\beta}=0 $.
	\par The condition for decomposable with respect to $g^{+}$, $\delta\delta^{\star}\omega=0$, generates the wave equation $\square^{+} X_{\beta}=-\nabla_{\beta}^{+}\nabla_{\mu}^{+}X^{\mu}-R^{+\lambda}_{\beta}$, for which solutions exist. 
	\par The divergence of $ v_{\beta}^{\alpha} $ in the mixed tensor bundle can be written as $ \nabla_{\alpha}v_{\beta}^{\alpha}=\partial_{\alpha}v_{\beta}^{\alpha}+\frac{v_{\beta}^{\alpha}}{2g}\partial_{\alpha}g-\frac{1}{2}v^{\alpha\lambda}\partial_{\beta}g_{\alpha\lambda} $ where $g=det(g_{\alpha\beta})$. This generates the (0,1) tensor
	\begin{equation}\label{Dv}
		\begin{split}
			\nabla_{\alpha}v_{\beta}^{\alpha}-\nabla^{+}_{\alpha}v_{\beta}^{\alpha}=v_{\beta}^{\alpha}\partial_{\alpha}(\ln\sqrt{-g}-\ln\sqrt{g^{+}})+v^{\alpha\lambda}\partial_{\beta}(u_{\alpha}u_{\lambda})
		\end{split}
	\end{equation} using (\ref{gab}). The metric compatibility of $ g_{\alpha\lambda}$ demands $v^{\alpha\lambda}\nabla_{\beta}g_{\alpha\lambda}=0$. In Riemann normal coordinates, $v^{\alpha\lambda}\partial_{\beta}(u_{\alpha}u_{\lambda})=0 $ holds for all $v^{\alpha\lambda} $, and $g=-g^{+}=-1$. The right-hand side of (\ref{Dv}) vanishes, which implies  $\nabla_{\alpha}v_{\beta}^{\alpha}-\nabla^{+}_{\alpha}v_{\beta}^{\alpha}=0 $ in all coordinate systems. Thus, $\nabla_{\alpha}v^{\alpha}_{\beta}=0 $ since $\nabla^{+}_{\alpha}v_{\beta}^{\alpha}=0 $ and
	\begin{equation}\label{decomp}
		\gamma_{\alpha\beta}=v_{\alpha\beta}+\frac{1}{2}{\mathcal L}_Xg_{\alpha\beta}+{\mathcal L}_X(u_{\alpha}u_{\beta})
	\end{equation}with $\nabla^{\alpha}v_{\alpha\beta}=0 $. Using $ X^{\lambda}=fu^{\lambda} $ where $f\neq0  $ is the magnitude of $X^{\lambda}  $, the expression  $X^{\lambda}\nabla_{\lambda}(u_{\alpha}u_{\beta}) $ in the second term of (\ref{decomp}) then vanishes in an affine parameterization and $
	\gamma_{\alpha\beta}=v_{\alpha\beta}+\varPhi_{\alpha\beta}	$ where $			\varPhi_{\alpha\beta}:=\frac{1}{2}(\nabla_{\alpha}X_{\beta}+\nabla_{\beta}X_{\alpha})+u^{\lambda}(u_{\alpha}\nabla_{\beta}X_{\lambda}+u_{\beta}\nabla_{\alpha}X_{\lambda})$.
	Provided $ \gamma_{\alpha\beta}\neq0 $, the decomposition is orthogonal: $<v_{\alpha\beta},\varPhi_{\alpha\beta}>=0  $. Since $\bm X$ is not a Killing vector, $\varPhi_{\alpha\beta}$ does not vanish.
	
\end{proof}
\section {Variations of the action functional}
There are three variables in MGR: $g^{\alpha\beta}$, $ X^{\beta} $ and $ u^{\beta} $. However, since $\bm{X}$ and $\bm{u}$ are collinear, $ X^{\beta}=fu^{\beta} $ and $X_{\beta}=fu_{\beta}  $ where $ f\neq0 $ is the magnitude of  both $ \bm{X} $ and its covector. Strictly speaking, the independent variables are $g^{+\alpha\beta}$, $ u^{\beta} $ and $ f $. However, since $ \frac{\delta }{\delta g^{+\alpha\beta}}=\frac{\delta }{\delta g^{\mu\nu}}\frac{\delta g^{\mu\nu}}{\delta g^{+\alpha\beta}}=\frac{\delta }{\delta g^{\alpha\beta}} $, $ g^{\alpha\beta} $ can be treated as an independent variable in the variation with respect to the inverse metric. Variations of $ S $ with respect to $g^{\alpha\beta}$, $ u^{\beta} $ and $ f $ are developed as follows:
\subsection*{Variation of $S^{G}$ with respect to $g^{\alpha\beta}$}
Variation of $ S^{EH} $ is well known from any textbook on GR. What needs to be established is the variation of $ S^{G} $ with respect to the inverse metric where 
$ S^{G}=-a\int\Phi\sqrt{-g}d^{4} $. The parameter $ a $ is not needed in this calculation, but does belong in $ S $. $ \Phi $ can be expressed as $\Phi=\nabla_{\alpha}X_{\beta}(g^{\alpha\beta}+2u^{\alpha}u^{\beta})=\nabla_{\alpha}X_{\beta}g^{+\alpha\beta}  $ so
\begin{equation}\label{dgPhi}
	\begin{split}
		-\delta \int\Phi\sqrt{-g}d^{4}x
		=-\int\delta (\nabla_{\alpha}X_{\beta})g^{+\alpha\beta}\sqrt{-g}d^{4}x-\int\nabla_{\alpha}X_{\beta}(\delta g^{\alpha\beta}+2\delta (u^{\alpha}u^{\beta}))\sqrt{-g}d^{4}x\\
		-\int\nabla_{\alpha}X_{\beta}(g^{\alpha\beta}+2u^{\alpha}u^{\beta})\delta\sqrt{-g}d^{4}x
	\end{split}
\end{equation}
The third integral in (\ref{dgPhi}) is equivalent to
\begin{equation}\label{3}
	\begin{split}
		\frac{1}{2}\int\nabla_{\mu}X_{\nu}(g^{\mu\nu}+2u^{\mu}u^{\nu})g_{\alpha\beta}\delta g^{\alpha\beta}\sqrt{-g}d^{4}x.
	\end{split}
\end{equation} \par To compute the second integral in (\ref{dgPhi}), consider $g_{\alpha\lambda} u^{\alpha}u^{\beta}=u_{\lambda}u^{\beta} $. In a Lorentzian spacetime $ \bm{u} $ satisfies $ u_{\lambda}u^{\lambda}=-1 $ but in a Riemannian spacetime, $ u_{\lambda}u^{\lambda}=1 $. Then $ \delta(u_{\lambda}u^{\beta})=0 $ 
\footnote{In $ g^{+}_{\alpha\beta}$, $u_{\lambda}u^{\lambda}=1  $ so in $ g_{\alpha\beta} $: $u_{\lambda}u^{\beta}=g^{\beta\mu}u_{\mu}u_{\lambda}\\=(g^{+\beta\mu}-2u^{\beta}u^{\mu})u_{\lambda}u_{\mu}=u_{\lambda}u^{\beta}-2u_{\lambda}u^{\beta}=-u_{\lambda}u^{\beta}$. Thus, if $\lambda=\beta, -1=-1 $ and if $\lambda\neq\beta, u_{\lambda}u^{\beta}=0  $ and $ \delta(u_{\lambda}u^{\beta})=0 $.} and $ \delta(u^{\alpha}u^{\beta})g_{\alpha\lambda}=-u^{\alpha}u^{\beta}\delta g_{\alpha\lambda} $.
Contracting that with $ g^{\lambda\rho} $ and re-labelling the indicies gives
\begin{equation}\label{deltauu}
	\delta(u^{\alpha}u^{\beta})=u^{\beta}u_{\lambda}\delta g^{\alpha\lambda}.
\end{equation}The second integral can then be expressed as $ -\int\nabla_{\alpha}X_{\beta}(\delta g^{\alpha\beta}+2u^{\beta}u_{\lambda}\delta g^{\alpha\lambda})\sqrt{-g}d^{4}x $. Keeping $ \alpha $ and $ \beta $ fixed while summing over $ \lambda $ in the integrand: if $ \lambda=\beta$, $ \nabla_{\alpha}X_{\beta}(\delta g^{\alpha\beta}+2u^{\lambda}u_{\lambda}\delta g^{\alpha\beta})=-\nabla_{\alpha}X_{\beta}\delta g^{\alpha\beta}$, and when $ \lambda\neq\beta $ the integrand is $\nabla_{\alpha}X_{\beta}\delta g^{\alpha\beta}  $ so the integrand vanishes in the sum over all $ \lambda $.
\par The first integral in (\ref{dgPhi}) is now calculated.
\begin{equation}\label{1}
	\begin{split}
		-\int\delta(\nabla_{\alpha}X_{\beta})g^{+\alpha\beta}\sqrt{-g}d^{4}x=-\int(\nabla_{\alpha}\delta X_{\beta}-X_{\lambda}\delta\Gamma^{\lambda}_{\alpha\beta})g^{+\alpha\beta}\sqrt{-g}d^{4}x\\
		=2\int\nabla_{\alpha}(u^{\alpha}u^{\beta})\delta X_{\beta}\sqrt{-g}d^{4}x+\int X_{\lambda}\delta\Gamma^{\lambda}_{\alpha\beta}g^{+\alpha\beta}\sqrt{-g}d^{4}x
	\end{split}
\end{equation}integrating by parts. The second integral in (\ref{1}) is expressed as
\begin{equation}\label{I2}
	\frac{1}{2}\int X_{\lambda}g^{+\alpha\beta}g^{\lambda\rho}(\nabla_{\alpha}\delta g_{\rho\beta}+\nabla_{\beta}\delta g_{\rho\alpha}-\nabla_{\rho}\delta g_{\alpha\beta})\sqrt{-g}d^{4}x.
\end{equation}To compute this integral, we need to use:\newline $ \delta(g^{+\alpha\beta}g_{\alpha\lambda})=\delta(\delta^{\beta}_{\lambda}+2u_{\lambda}u^{\beta})=0 $ so $ \delta g^{+\alpha\beta}g_{\alpha\lambda}=-g^{+\alpha\beta}\delta g_{\alpha\lambda} $. Contracting that with $ g^{\lambda\rho} $ gives
\begin{equation}\label{g+1}
	\delta g^{+\rho\beta}=-g^{+\alpha\beta}g^{\lambda\rho}\delta g_{\alpha\lambda}.
\end{equation}
Continuing with (\ref{I2}), we can use the Lorentz connection $ \nabla $ with $ g^{+} $ in the first two terms because it can operate on the divergence of a symmetric tensor as shown in the ODT:
\begin{equation}\
	\begin{split}
		-\frac{1}{2}\int X_{\lambda}
		(\nabla_{\alpha}\delta g^{+\alpha\lambda}+\nabla_{\beta}\delta g^{+\lambda\beta})\sqrt{-g}d^{4}x-\frac{1}{2}\int X_{\lambda}(g^{\alpha\beta}+2u^{\alpha}u^{\beta})g^{\lambda\rho}\nabla_{\rho}\delta g_{\alpha\beta}\sqrt{-g}d^{4}x\\
		=-\int X_{\lambda}
		\nabla_{\alpha}\delta g^{+\alpha\lambda}\sqrt{-g}d^{4}x+\frac{1}{2}\int X_{\lambda}g_{\alpha\beta}\nabla^{\lambda}\delta g^{\alpha\beta}\sqrt{-g}d^{4}x-\int X_{\lambda}u^{\alpha}u^{\beta}\nabla^{\lambda}\delta g_{\alpha\beta}\sqrt{-g}d^{4}x\\
		=-\int X_{\beta}
		\nabla_{\alpha}\delta (g^{\alpha\beta}+2u^{\alpha}u^{\beta})\sqrt{-g}d^{4}x+\frac{1}{2}\int X_{\lambda}g_{\alpha\beta}\nabla^{\lambda}\delta g^{\alpha\beta}\sqrt{-g}d^{4}x\\+\int X_{\lambda}u_{\alpha}u_{\beta}\nabla^{\lambda}\delta g^{\alpha\beta}\sqrt{-g}d^{4}x
	\end{split}
\end{equation}using $ u^{\alpha}u^{\beta}\nabla^{\lambda}\delta g_{\alpha\beta}=u_{\mu}u_{\nu}\nabla^{\lambda}g^{\alpha\mu}g^{\beta\nu}\delta g_{\alpha\beta}=-u_{\alpha}u_{\beta}\nabla^{\lambda}\delta g^{\alpha\beta} $. Integrating by parts gives
\begin{equation}
	\begin{split}
		\int [\nabla_{\alpha}X_{\beta}-\frac{1}{2}\nabla^{\lambda}X_{\lambda}g_{\alpha\beta}-\nabla^{\lambda}(X_{\lambda}u_{\alpha}u_{\beta})]\delta g^{\alpha\beta}\sqrt{-g}d^{4}x+2\int \nabla_{\alpha}X_{\beta}\delta (u^{\alpha}u^{\beta})\sqrt{-g}d^{4}x\\
		=\int [\nabla_{\alpha}X_{\beta}-\frac{1}{2}\nabla^{\lambda}X_{\lambda}g_{\alpha\beta}-\nabla^{\lambda}(X_{\lambda}u_{\alpha}u_{\beta})+2u^{\lambda}u_{\beta}\nabla_{\alpha}X_{\lambda}]\delta g^{\alpha\beta}\sqrt{-g}d^{4}x
	\end{split}
\end{equation}from (\ref{deltauu}) and re-labelling indicies. \par The complete variation is then
\begin{equation}
	\begin{split}
		\int [\nabla_{\alpha}X_{\beta}-\frac{1}{2}\nabla^{\lambda}X_{\lambda}g_{\alpha\beta}-\nabla^{\lambda}(X_{\lambda}u_{\alpha}u_{\beta})+2u^{\lambda}u_{\beta}\nabla_{\alpha}X_{\lambda}\\+\frac{1}{2}\nabla_{\mu}X_{\nu}(g^{\mu\nu}+2u^{\mu}u^{\nu})g_{\alpha\beta}]\delta g^{\alpha\beta}\sqrt{-g}d^{4}x\\
		+2\int\nabla_{\alpha}(u^{\alpha}u^{\beta})\delta X_{\beta}\sqrt{-g}d^{4}x.
	\end{split}
\end{equation}Since the variations in $ \delta g^{\alpha\beta} $ are independent to those of $ \delta X_{\beta} $, the variation in S will yield
\begin{equation}\label{nuabapp}
	\nabla_{\alpha}(u^{\alpha}u^{\beta})=0
\end{equation} and the contribution from the variation of $ S^{G} $ with respect to $ g^{\alpha\beta} $ is
\begin{equation}
	\begin{split}
		\int[\nabla_{\alpha}X_{\beta}+2u^{\lambda}u_{\beta}\nabla_{\alpha}X_{\lambda}-\nabla^{\lambda}(X_{\lambda}u_{\alpha}u_{\beta})+\nabla_{\mu}X_{\nu}u^{\mu}u^{\nu}g_{\alpha\beta}]\delta g^{\alpha\beta}\sqrt{-g}d^{4}x\\
		=\int[\nabla_{\alpha}X_{\beta}+2u^{\lambda}u_{\beta}\nabla_{\alpha}X_{\lambda}-X^{\lambda}\nabla_{\lambda}(u_{\alpha}u_{\beta})+\nabla_{\mu}X_{\nu}(u^{\mu}u^{\nu}g_{\alpha\beta}-u_{\alpha}u_{\beta}g^{\mu\nu})]\delta g^{\alpha\beta}\sqrt{-g}d^{4}x\\
		=\int[\frac{1}{2}(\nabla_{\alpha}X_{\beta}+\nabla_{\beta}X_{\alpha})+u^{\lambda}(u_{\alpha}\nabla_{\beta}X_{\lambda}+u_{\beta}\nabla_{\alpha}X_{\lambda})\\+\nabla_{\mu}X_{\nu}(u^{\mu}u^{\nu}g_{\alpha\beta}-u_{\alpha}u_{\beta}g^{\mu\nu})]\delta g^{\alpha\beta}\sqrt{-g}d^{4}x
	\end{split}
\end{equation}using the symmetry in $ \delta g^{\alpha\beta} $ and setting $-X^{\lambda}\nabla_{\lambda}(u_{\alpha}u_{\beta})=-fu^{\lambda}\nabla_{\lambda}(u_{\alpha}u_{\beta})=0  $ in an affine parameterization where $ f\neq0 $ is the magnitude of $ X^{\lambda} $. The last term in the variation vanishes, which follows by writing the tensor in brackets   $-u_{\alpha}u_{\beta}g^{\mu\nu}+u^{\mu}u^{\nu}g_{\alpha\beta}$ as its equivalent, $\frac{1}{2}(g^{{+}{\mu\nu}}g_{\alpha\beta}-g^{+}_{\alpha\beta}g^{\mu\nu}) $, and choosing an orthonormal basis $(e_{\alpha})$ at a point $p\in\mathcal{M} $ for $ g^{+} $ with $ e_{0}=u $. Then, $ u^{0}=u_{0}=1 $, $ u^{i}=u_{i}=0\;(i=1,2,3)$, $g^{+}_{\alpha\beta}=\delta_{\alpha\beta}$, $ g^{00}=-g^{{+}00}$ and $g_{00}=-g^{+}_{00}$, with all other components of the metric g equal to those of the metric $ g^{+}$. From the definition of $\varPhi_{\alpha\beta}  $
\begin{equation}\label{}
	\varPhi_{\alpha\beta}:=\frac{1}{2}(\nabla_{\alpha}X_{\beta}+\nabla_{\beta}X_{\alpha})+u^{\lambda}(u_{\alpha}\nabla_{\beta}X_{\lambda}+u_{\beta}\nabla_{\alpha}X_{\lambda})
\end{equation}it follows that 
$ \delta{S}^{G}=-a\delta \int\Phi\sqrt{-g}d^{4}x=a\int\varPhi_{\alpha\beta}\delta g^{\alpha\beta}\sqrt{-g}d^{4}x.$
\subsection*{Variation of $S$ with respect to $f$}
$ f\neq0 $, the magnitude of $ X^{\beta} $, is independent of $ u^{\beta} $ so variation of $ S $ with respect to $ f $ must be performed. With $ X^{\beta} $ and its first derivative not appearing in the matter energy-momentum tensor, variation with respect to $ f $ involves only $ S^{G} $:
\begin{equation}
	\begin{split}
		S^{G}=-a\int[\nabla_{\alpha}(fu^{\alpha})+2u^{\alpha}u^{\beta}\nabla_{\alpha}(fu_{\beta})]\sqrt{-g}d^{4}x\\
		=-a\int(f\nabla_{\alpha}u^{\alpha}-u^{\alpha}\nabla_{\alpha}f)\sqrt{-g}d^{4}x
	\end{split}
\end{equation}
using $ X_{\alpha}=fu_{\alpha}$ and $u^{\beta}\nabla_{\alpha}u_{\beta}=0$. Hence, 
\begin{equation}
	\begin{split}
		\delta S^{G}&=-a\int(\nabla_{\alpha}u^{\alpha}\delta f-u^{\alpha}\delta\nabla_{\alpha}f)\sqrt{-g}d^{4}x\\
		&=-2a\int\nabla_{\alpha}u^{\alpha}\delta f\sqrt{-g}d^{4}x.
	\end{split}
\end{equation} It follows that 
\begin{equation}\label{bnu0}
	\nabla_{\alpha}u^{\alpha}=0 
\end{equation} for arbitrary variations of f. $ \nabla_{\alpha}u^{\alpha}=0 $ can also be obtained from (\ref{uab}) in an affine parameterization.
\subsection*{Variation of $S$ with respect to $u^{\nu}$} 
The action functionals $S^{F}$ and $S^{EH}$ do not depend on $X^{\mu}$ and  $\frac{\delta X^{\mu}}{\delta u^{\nu}}=f\delta^{\mu}_{\nu}$ is well defined, so the variation with respect to $u^{\nu}$ of $S$ depends only on $S^{G}$:
\begin{equation}\label{Su}
	\frac{\delta S^{G}}{\delta u^{\nu}}=-\int\frac{\delta }{\delta u^{\nu}}(\Phi\sqrt{-g})d^{4}x.
\end{equation}
From the collinearity $X_{\alpha}=fu_{\alpha}$, $\varPhi_{\alpha\beta}$ can be expressed as
\begin{equation}\label{phiu}
	\varPhi_{\alpha\beta}=\frac{f}{2}(\nabla_{\alpha}u_{\beta}+\nabla_{\beta}u_{\alpha})-\frac{1}{2}(u_{\beta}\nabla_{\alpha}f+u_{\alpha}\nabla_{\beta}f).
\end{equation}Using (\ref{bnu0}), $\Phi$ can be expressed as
\begin{equation}\label{Phiu}
	\Phi=-u^{\alpha}\nabla_{\alpha}f.
\end{equation}
Then, (\ref{Su}) becomes
\begin{equation}\label{dSu}
	\frac{\delta S^{G}}{\delta u^{\nu}}=\int[(\delta^{\alpha}_{\nu}\partial_{\alpha}f+u^{\alpha}\frac{\delta}{\delta u^{\nu}}(\partial_{\alpha}f) )\sqrt{-g}-\Phi\frac{\delta\sqrt{-g}}{\delta u^{\nu}}]d^{4}x
\end{equation} The third term is
\begin{equation}\label{sqrtg}
	\begin{split}
		-\Phi\frac{\delta\sqrt{-g}}{\delta u^{\nu}}&=-\Phi\frac{\delta\sqrt{-g}}{\delta g^{\alpha\beta}}\frac{\delta g^{\alpha\beta}}{\delta u^{\nu}}\\
		&=-\Phi g_{\alpha\beta}\sqrt{-g}(\delta^{\alpha}_{\nu}u^{\beta}+\delta^{\beta}_{\nu}u^{\alpha})\\
		&=-2\Phi\sqrt{-g}u_{\nu}
	\end{split}
\end{equation}so (\ref{dSu}) can be written as
\begin{equation}
	\frac{\delta S^{G}}{\delta u^{\nu}}=\int(\partial_{\nu}f+u^{\alpha}\frac{\delta}{\delta u^{\nu}}(\partial_{\alpha}f)-2\Phi u_{\nu})\sqrt{-g})d^{4}x
\end{equation}

Setting $\partial_{\alpha}f=u_{\alpha}P$ for some scalar $P$ requires
\begin{equation}
	u^{\alpha}\frac{\delta}{\delta u^{\nu}}(\partial_{\alpha}f)=-\frac{\delta P}{\delta u^{\nu}}-P u_{\nu}
\end{equation}using 
\begin{equation}
	u^{\alpha}\frac{\delta u_{\alpha}}{\delta u^{\nu}}=-u_{\nu}
\end{equation} from $u_{\alpha}u^{\alpha}=-1$. Thus, 
\begin{equation}
	\frac{\delta S^{G}}{\delta u^{\nu}}=-\int(\frac{\delta \Phi}{\delta u^{\nu}}\sqrt{-g}+\Phi \frac{\delta\sqrt{-g}}{\delta u^{\nu}})d^{4}x=\int(\partial_{\nu}f-\frac{\delta P}{\delta u^{\nu}}-(P+2\Phi) u_{\nu})\sqrt{-g})d^{4}x=0.
\end{equation}Choosing $P=\Phi$ generates
\begin{equation}\label{bunu}
	\partial_{\nu}f=\Phi u_{\nu}.
\end{equation}
\par Note that (\ref{bnu0}) and (\ref{bunu}) are consistent with $ \Phi=\nabla_{\alpha}X^{\alpha}+2u^{\alpha}u^{\beta}\nabla_{\alpha}X_{\beta} $:\begin{equation}\label{PhiPhi}
	\begin{split}
		\Phi&=\nabla_{\alpha}X^{\alpha}+2u^{\alpha}u^{\beta}\nabla_{\alpha}X_{\beta}\\
		&=f\nabla_{\alpha}u^{\alpha}+u^{\alpha}\partial_{\alpha}f+2u^{\alpha}u^{\beta}u_{\beta}\partial_{\alpha}f+2fu^{\alpha}u^{\beta}\nabla_{\alpha}u_{\beta}\\
		&=u^{\alpha}\partial_{\alpha}f-2u^{\alpha}\partial_{\alpha}f\\
		&=-u^{\alpha}u_{\alpha}\Phi\\
		&=\Phi.
	\end{split}
\end{equation}

\section*{Table 2: Calculation of $\alpha\gamma$ for 163 SPARC disc galaxies}
Format: Galaxy, $R_{p}$(kpc), $V_{b}$(km/s), $M_{b}\times10^{9}M\odot$,	$R_{p}$(m),	$V_{b}$(m/s), $M_{b}$(kg), $\alpha\gamma=\frac{V^{4}_{b}R^{2}_{p}}{G^{2}M^{2}_{b}}$

		\begin{longtable}{|l|l|l|l|l|l|l|l|}
		\hline
		D512-2 & 1.92 & 33.24 & 0.77 & 5.925E+19 & 3.324E+04 & 1.532E+39 & 41.05 \\ \hline
		D564-8 & 1.53 & 15.364 & 0.08 & 4.721E+19 & 1.536E+04 & 1.591E+38 & 110.23 \\ \hline
		D631-7 & 1.8 & 27.771 & 0.18 & 5.554E+19 & 2.777E+04 & 3.580E+38 & 321.69 \\ \hline
		DDO064 & 2.08 & 34.414 & 0.34 & 6.418E+19 & 3.441E+04 & 6.763E+38 & 283.91 \\ \hline
		DDO154 & 0.99 & 30.997 & 0.23 & 3.055E+19 & 3.100E+04 & 4.575E+38 & 92.51 \\ \hline
		DDO170 & 7.85 & 46.874 & 2.28 & 2.422E+20 & 4.687E+04 & 4.535E+39 & 309.51 \\ \hline
		ESO079-G014 & 13.36 & 136.226 & 50.45 & 4.122E+20 & 1.362E+05 & 1.003E+41 & 130.62 \\ \hline
		ESO116-G012 & 4.13 & 76.318 & 5.2 & 1.274E+20 & 7.632E+04 & 1.034E+40 & 115.74 \\ \hline
		ESO444-G084 & 1.29 & 37.552 & 0.38 & 3.981E+19 & 3.755E+04 & 7.558E+38 & 123.94 \\ \hline
		ESO563-G021 & 16.83 & 231.148 & 206.22 & 5.193E+20 & 2.311E+05 & 4.102E+41 & 102.84 \\ \hline
		F561-1 & 9.66 & 54.449 & 3.52 & 2.981E+20 & 5.445E+04 & 7.001E+39 & 358.02 \\ \hline
		F563-V1 & 7.87 & 30.956 & 1.56 & 2.428E+20 & 3.096E+04 & 3.103E+39 & 126.40 \\ \hline
		F563-V2 & 4.76 & 93.395 & 12.42 & 1.469E+20 & 9.340E+04 & 2.470E+40 & 60.44 \\ \hline
		F565-V2 & 7.54 & 54.901 & 3.19 & 2.327E+20 & 5.490E+04 & 6.345E+39 & 274.51 \\ \hline
		F567-2 & 9.59 & 50.888 & 3.91 & 2.959E+20 & 5.089E+04 & 7.777E+39 & 218.18 \\ \hline
		F568-1 & 13.23 & 98.873 & 18.86 & 4.082E+20 & 9.887E+04 & 3.751E+40 & 254.35 \\ \hline
		F568-3 & 15.94 & 69.172 & 11.3 & 4.919E+20 & 6.917E+04 & 2.248E+40 & 246.39 \\ \hline
		F568-V1 & 8.21 & 99.153 & 17.05 & 2.533E+20 & 9.915E+04 & 3.391E+40 & 121.21 \\ \hline
		F571-8 & 2.33 & 56.064 & 2.42 & 7.190E+19 & 5.606E+04 & 4.813E+39 & 49.53 \\ \hline
		F571-V1 & 9.72 & 59.084 & 4.74 & 2.999E+20 & 5.908E+04 & 9.428E+39 & 277.16 \\ \hline
		F574-1 & 12.6 & 77.98 & 13.44 & 3.888E+20 & 7.798E+04 & 2.673E+40 & 175.77 \\ \hline
		F574-2 & 10.83 & 39.279 & 1.63 & 3.342E+20 & 3.928E+04 & 3.242E+39 & 568.33 \\ \hline
		F579-V1 & 10.85 & 120.955 & 30.67 & 3.348E+20 & 1.210E+05 & 6.100E+40 & 144.88 \\ \hline
		F583-4 & 6.74 & 56.772 & 3.37 & 2.080E+20 & 5.677E+04 & 6.703E+39 & 224.74 \\ \hline
		IC4202 & 11.79 & 151.763 & 68.44 & 3.638E+20 & 1.518E+05 & 1.361E+41 & 85.14 \\ \hline
		KK98-251 & 2.79 & 27.948 & 0.23 & 8.609E+19 & 2.795E+04 & 4.575E+38 & 485.55 \\ \hline
		NGC0024 & 2.65 & 93.812 & 6.8 & 8.177E+19 & 9.381E+04 & 1.353E+40 & 63.62 \\ \hline
		NGC0055 & 13.5 & 61.912 & 5.43 & 4.166E+20 & 6.191E+04 & 1.080E+40 & 491.18 \\ \hline
		NGC0100 & 3.88 & 56.887 & 3.29 & 1.197E+20 & 5.689E+04 & 6.544E+39 & 78.78 \\ \hline
		NGC0247 & 10.76 & 89.827 & 16.57 & 3.320E+20 & 8.983E+04 & 3.296E+40 & 148.49 \\ \hline
		NGC0289 & 2.52 & 201.166 & 56.81 & 7.776E+19 & 2.012E+05 & 1.130E+41 & 17.43 \\ \hline
		NGC0300 & 4.54 & 65.226 & 4.73 & 1.401E+20 & 6.523E+04 & 9.408E+39 & 90.19 \\ \hline
		NGC0801 & 3.52 & 201.046 & 159.16 & 1.086E+20 & 2.010E+05 & 3.166E+41 & 4.32 \\ \hline
		NGC0891 & 3.24 & 236.359 & 59.02 & 9.998E+19 & 2.364E+05 & 1.174E+41 & 50.87 \\ \hline
		NGC1003 & 5.39 & 83.71 & 7.18 & 1.663E+20 & 8.371E+04 & 1.428E+40 & 149.66 \\ \hline
		NGC1090 & 8.17 & 133.907 & 43.81 & 2.521E+20 & 1.339E+05 & 8.714E+40 & 60.48 \\ \hline
		NGC1705 & 0.66 & 73.488 & 1.38 & 2.037E+19 & 7.349E+04 & 2.745E+39 & 36.08 \\ \hline
		NGC2366 & 5.11 & 30.269 & 0.23 & 1.577E+20 & 3.027E+04 & 4.575E+38 & 2241.07 \\ \hline
		NGC2403 & 2.28 & 100.233 & 8.81 & 7.035E+19 & 1.002E+05 & 1.752E+40 & 36.56 \\ \hline
		NGC2683 & 5.77 & 213.846 & 64.35 & 1.780E+20 & 2.138E+05 & 1.280E+41 & 90.94 \\ \hline
		NGC2841 & 3.44 & 313.632 & 220.13 & 1.061E+20 & 3.136E+05 & 4.378E+41 & 12.78 \\ \hline
		NGC2903 & 3.2 & 203.626 & 37.35 & 9.874E+19 & 2.036E+05 & 7.429E+40 & 68.25 \\ \hline
		NGC2915 & 0.67 & 36.051 & 0.33 & 2.067E+19 & 3.605E+04 & 6.564E+38 & 37.66 \\ \hline
		NGC2955 & 4.61 & 254.451 & 143.32 & 1.423E+20 & 2.545E+05 & 2.851E+41 & 23.46 \\ \hline
		NGC2976 & 1.74 & 74.321 & 2.42 & 5.369E+19 & 7.432E+04 & 4.813E+39 & 85.31 \\ \hline
		NGC2998 & 7.59 & 230.857 & 162.93 & 2.342E+20 & 2.309E+05 & 3.241E+41 & 33.34 \\ \hline
		NGC3109 & 3.61 & 36.801 & 0.9 & 1.114E+20 & 3.680E+04 & 1.790E+39 & 159.60 \\ \hline
		NGC3198 & 5.84 & 134.597 & 33.98 & 1.802E+20 & 1.346E+05 & 6.759E+40 & 52.43 \\ \hline
		NGC3521 & 4.12 & 217.115 & 49.85 & 1.271E+20 & 2.171E+05 & 9.915E+40 & 82.09 \\ \hline
		NGC3726 & 8.72 & 149.228 & 45.65 & 2.691E+20 & 1.492E+05 & 9.080E+40 & 97.87 \\ \hline
		NGC3769 & 3.5 & 113.072 & 12.38 & 1.080E+20 & 1.131E+05 & 2.462E+40 & 70.66 \\ \hline
		NGC3877 & 9.6 & 154.123 & 40.79 & 2.962E+20 & 1.541E+05 & 8.113E+40 & 169.04 \\ \hline
		NGC3893 & 5.24 & 176.129 & 39.59 & 1.617E+20 & 1.761E+05 & 7.874E+40 & 91.18 \\ \hline
		NGC3917 & 7.85 & 121.605 & 26.36 & 2.422E+20 & 1.216E+05 & 5.243E+40 & 104.89 \\ \hline
		NGC3949 & 3.5 & 154.94 & 22.35 & 1.080E+20 & 1.549E+05 & 4.445E+40 & 76.44 \\ \hline
		NGC3953 & 10.47 & 224.437 & 115.38 & 3.231E+20 & 2.244E+05 & 2.295E+41 & 113.00 \\ \hline
		NGC3972 & 6.11 & 118.902 & 17.14 & 1.885E+20 & 1.189E+05 & 3.409E+40 & 137.37 \\ \hline
		NGC3992 & 13.79 & 274.1 & 280.12 & 4.255E+20 & 2.741E+05 & 5.572E+41 & 73.99 \\ \hline
		NGC4010 & 7.85 & 91.953 & 12.83 & 2.422E+20 & 9.195E+04 & 2.552E+40 & 144.75 \\ \hline
		NGC4013 & 7.69 & 197.205 & 54.92 & 2.373E+20 & 1.972E+05 & 1.092E+41 & 160.38 \\ \hline
		NGC4051 & 8.72 & 161 & 54.92 & 2.691E+20 & 1.610E+05 & 1.092E+41 & 91.61 \\ \hline
		NGC4085 & 4.45 & 97.751 & 8.12 & 1.373E+20 & 9.775E+04 & 1.615E+40 & 148.31 \\ \hline
		NGC4088 & 6.98 & 179 & 47.74 & 2.154E+20 & 1.790E+05 & 9.495E+40 & 118.70 \\ \hline
		NGC4100 & 6.98 & 202.677 & 56.27 & 2.154E+20 & 2.027E+05 & 1.119E+41 & 140.43 \\ \hline
		NGC4138 & 2.61 & 208.959 & 36.28 & 8.054E+19 & 2.090E+05 & 7.216E+40 & 53.37 \\ \hline
		NGC4157 & 8.72 & 196.216 & 62 & 2.691E+20 & 1.962E+05 & 1.233E+41 & 158.59 \\ \hline
		NGC4183 & 8.72 & 113.82 & 19.63 & 2.691E+20 & 1.138E+05 & 3.904E+40 & 179.12 \\ \hline
		NGC4214 & 1.04 & 70.733 & 1.65 & 3.209E+19 & 7.073E+04 & 3.282E+39 & 53.79 \\ \hline
		NGC4217 & 3.5 & 152.979 & 24.88 & 1.080E+20 & 1.530E+05 & 4.949E+40 & 58.62 \\ \hline
		NGC4389 & 4.45 & 65.338 & 3.49 & 1.373E+20 & 6.534E+04 & 6.942E+39 & 160.25 \\ \hline
		NGC4559 & 4.58 & 104 & 15.07 & 1.413E+20 & 1.040E+05 & 2.997E+40 & 58.44 \\ \hline
		NGC5005 & 3.44 & 245.314 & 96.56 & 1.061E+20 & 2.453E+05 & 1.921E+41 & 24.86 \\ \hline
		NGC5033 & 2.28 & 258.067 & 78.78 & 7.035E+19 & 2.581E+05 & 1.567E+41 & 20.09 \\ \hline
		NGC5055 & 6.46 & 197.327 & 75.39 & 1.993E+20 & 1.973E+05 & 1.500E+41 & 60.21 \\ \hline
		NGC5371 & 14.48 & 246.665 & 227.56 & 4.468E+20 & 2.467E+05 & 4.526E+41 & 81.07 \\ \hline
		NGC5585 & 4.79 & 49.99 & 1.92 & 1.478E+20 & 4.999E+04 & 3.819E+39 & 210.22 \\ \hline
		NGC5907 & 12.58 & 241.545 & 187.91 & 3.882E+20 & 2.415E+05 & 3.738E+41 & 82.51 \\ \hline
		NGC5985 & 14.1 & 295.602 & 331.87 & 4.351E+20 & 2.956E+05 & 6.601E+41 & 74.54 \\ \hline
		NGC6015 & 4.6 & 152 & 32.97 & 1.419E+20 & 1.520E+05 & 6.558E+40 & 56.20 \\ \hline
		NGC6195 & 6.34 & 233.455 & 184.17 & 1.956E+20 & 2.335E+05 & 3.663E+41 & 19.04 \\ \hline
		NGC6503 & 2.28 & 130.425 & 10.68 & 7.035E+19 & 1.304E+05 & 2.124E+40 & 71.33 \\ \hline
		NGC6674 & 12.54 & 297.42 & 312.79 & 3.869E+20 & 2.974E+05 & 6.221E+41 & 68.02 \\ \hline
		NGC6789 & 0.71 & 38.604 & 0.3 & 2.191E+19 & 3.860E+04 & 5.967E+38 & 67.28 \\ \hline
		NGC6946 & 0.19 & 180.511 & 39.23 & 5.863E+18 & 1.805E+05 & 7.803E+40 & 0.13 \\ \hline
		NGC7331 & 4.27 & 257.417 & 112.65 & 1.318E+20 & 2.574E+05 & 2.241E+41 & 34.12 \\ \hline
		NGC7793 & 3.16 & 74.515 & 3.75 & 9.751E+19 & 7.452E+04 & 7.459E+39 & 118.40 \\ \hline
		NGC7814 & 0.63 & 303.055 & 56.69 & 1.944E+19 & 3.031E+05 & 1.128E+41 & 5.63 \\ \hline
		PGC51017 & 1.65 & 19.937 & 0.09 & 5.091E+19 & 1.994E+04 & 1.790E+38 & 287.21 \\ \hline
		UGC00128 & 16.22 & 118.742 & 48.89 & 5.005E+20 & 1.187E+05 & 9.724E+40 & 118.35 \\ \hline
		UGC00191 & 4.97 & 61.719 & 3.76 & 1.534E+20 & 6.172E+04 & 7.479E+39 & 137.12 \\ \hline
		UGC00634 & 4.51 & 98.707 & 14.72 & 1.392E+20 & 9.871E+04 & 2.928E+40 & 48.20 \\ \hline
		UGC00731 & 8.19 & 66.746 & 4.94 & 2.527E+20 & 6.675E+04 & 9.826E+39 & 295.05 \\ \hline
		UGC00891 & 4.45 & 44.487 & 1.52 & 1.373E+20 & 4.449E+04 & 3.023E+39 & 181.57 \\ \hline
		UGC01230 & 8.6 & 76.964 & 14.49 & 2.654E+20 & 7.696E+04 & 2.882E+40 & 66.85 \\ \hline
		UGC01281 & 3.46 & 33.634 & 0.76 & 1.068E+20 & 3.363E+04 & 1.512E+39 & 143.46 \\ \hline
		UGC02023 & 3.78 & 37.545 & 1.01 & 1.166E+20 & 3.755E+04 & 2.009E+39 & 150.53 \\ \hline
		UGC02259 & 4.07 & 84.829 & 6.61 & 1.256E+20 & 8.483E+04 & 1.315E+40 & 106.18 \\ \hline
		UGC02455 & 4.03 & 37.736 & 0.46 & 1.244E+20 & 3.774E+04 & 9.149E+38 & 841.77 \\ \hline
		UGC02487 & 10.35 & 392.161 & 684.46 & 3.194E+20 & 3.922E+05 & 1.361E+42 & 29.25 \\ \hline
		UGC02885 & 6.82 & 266.248 & 359.41 & 2.104E+20 & 2.662E+05 & 7.149E+41 & 9.79 \\ \hline
		UGC02916 & 2.6 & 214.853 & 69.67 & 8.023E+19 & 2.149E+05 & 1.386E+41 & 16.05 \\ \hline
		UGC02953 & 7.2 & 272.586 & 169.58 & 2.222E+20 & 2.726E+05 & 3.373E+41 & 53.83 \\ \hline
		UGC03205 & 7.27 & 223.969 & 101.05 & 2.243E+20 & 2.240E+05 & 2.010E+41 & 70.44 \\ \hline
		UGC03546 & 0.63 & 284.649 & 51.84 & 1.944E+19 & 2.846E+05 & 1.031E+41 & 5.24 \\ \hline
		UGC03580 & 1.37 & 69.773 & 4.3 & 4.227E+19 & 6.977E+04 & 8.553E+39 & 13.01 \\ \hline
		UGC04305 & 5.02 & 35.2 & 0.45 & 1.549E+20 & 3.520E+04 & 8.951E+38 & 1033.31 \\ \hline
		UGC04325 & 3.49 & 95.414 & 6.88 & 1.077E+20 & 9.541E+04 & 1.368E+40 & 115.35 \\ \hline
		UGC04483 & 0.73 & 15.778 & 0.02 & 2.253E+19 & 1.578E+04 & 3.978E+37 & 446.55 \\ \hline
		UGC04499 & 4.55 & 53.352 & 2.47 & 1.404E+20 & 5.335E+04 & 4.913E+39 & 148.70 \\ \hline
		UGC05005 & 10.17 & 55.229 & 5.51 & 3.138E+20 & 5.523E+04 & 1.096E+40 & 171.43 \\ \hline
		UGC05253 & 3.68 & 244.654 & 121.58 & 1.136E+20 & 2.447E+05 & 2.418E+41 & 17.75 \\ \hline
		UGC05414 & 4.11 & 45.267 & 1.43 & 1.268E+20 & 4.527E+04 & 2.844E+39 & 187.59 \\ \hline
		UGC05716 & 2.07 & 56.259 & 2.4 & 6.387E+19 & 5.626E+04 & 4.774E+39 & 40.31 \\ \hline
		UGC05721 & 1.17 & 55.487 & 0.79 & 3.610E+19 & 5.549E+04 & 1.571E+39 & 112.45 \\ \hline
		UGC05750 & 9.43 & 55.705 & 4.9 & 2.910E+20 & 5.571E+04 & 9.746E+39 & 192.88 \\ \hline
		UGC05764 & 2.9 & 47.483 & 1.2 & 8.949E+19 & 4.748E+04 & 2.387E+39 & 160.57 \\ \hline
		UGC05829 & 6.91 & 53.77 & 2.33 & 2.132E+20 & 5.377E+04 & 4.634E+39 & 397.63 \\ \hline
		UGC05918 & 3.9 & 42.6 & 1.2 & 1.203E+20 & 4.260E+04 & 2.387E+39 & 188.14 \\ \hline
		UGC05986 & 5.02 & 92.006 & 7.25 & 1.549E+20 & 9.201E+04 & 1.442E+40 & 185.81 \\ \hline
		UGC06399 & 6.11 & 68 & 4.92 & 1.885E+20 & 6.800E+04 & 9.786E+39 & 178.35 \\ \hline
		UGC06446 & 3.49 & 64.567 & 3.75 & 1.077E+20 & 6.457E+04 & 7.459E+39 & 81.42 \\ \hline
		UGC06614 & 3.23 & 221.403 & 96.68 & 9.967E+19 & 2.214E+05 & 1.923E+41 & 14.51 \\ \hline
		UGC06628 & 7.69 & 50.444 & 2.67 & 2.373E+20 & 5.044E+04 & 5.311E+39 & 290.50 \\ \hline
		UGC06667 & 7.85 & 95.442 & 20.64 & 2.422E+20 & 9.544E+04 & 4.105E+40 & 64.92 \\ \hline
		UGC06786 & 1.3 & 193.671 & 62.08 & 4.011E+19 & 1.937E+05 & 1.235E+41 & 3.34 \\ \hline
		UGC06787 & 0.71 & 253.781 & 59.7 & 2.191E+19 & 2.538E+05 & 1.187E+41 & 3.17 \\ \hline
		UGC06917 & 10.47 & 81.355 & 9.92 & 3.231E+20 & 8.136E+04 & 1.973E+40 & 263.93 \\ \hline
		UGC06923 & 3.72 & 57.634 & 2.5 & 1.148E+20 & 5.763E+04 & 4.973E+39 & 132.13 \\ \hline
		UGC06930 & 10.47 & 97.495 & 16.47 & 3.231E+20 & 9.750E+04 & 3.276E+40 & 197.48 \\ \hline
		UGC06973 & 1.74 & 192.189 & 15.58 & 5.369E+19 & 1.922E+05 & 3.099E+40 & 92.04 \\ \hline
		UGC06983 & 2.61 & 87.661 & 10.2 & 8.054E+19 & 8.766E+04 & 2.029E+40 & 20.91 \\ \hline
		UGC07089 & 6.11 & 56.586 & 3.98 & 1.885E+20 & 5.659E+04 & 7.916E+39 & 130.69 \\ \hline
		UGC07125 & 13 & 53.314 & 3.52 & 4.011E+20 & 5.331E+04 & 7.001E+39 & 596.00 \\ \hline
		UGC07151 & 5.5 & 61.099 & 2.93 & 1.697E+20 & 6.110E+04 & 5.828E+39 & 265.59 \\ \hline
		UGC07232 & 0.62 & 28.361 & 0.12 & 1.913E+19 & 2.836E+04 & 2.387E+38 & 93.41 \\ \hline
		UGC07261 & 2.86 & 66.377 & 2.92 & 8.825E+19 & 6.638E+04 & 5.808E+39 & 100.72 \\ \hline
		UGC07323 & 5.23 & 68.446 & 4.25 & 1.614E+20 & 6.845E+04 & 8.453E+39 & 179.76 \\ \hline
		UGC07399 & 1.22 & 70.4 & 3.71 & 3.765E+19 & 7.040E+04 & 7.379E+39 & 14.37 \\ \hline
		UGC07524 & 10.35 & 66.036 & 5.12 & 3.194E+20 & 6.604E+04 & 1.018E+40 & 420.28 \\ \hline
		UGC07577 & 1.51 & 13.286 & 0.03 & 4.659E+19 & 1.329E+04 & 5.967E+37 & 426.94 \\ \hline
		UGC07603 & 1.71 & 40.98 & 0.55 & 5.277E+19 & 4.098E+04 & 1.094E+39 & 147.45 \\ \hline
		UGC07608 & 4.78 & 45.299 & 1.47 & 1.475E+20 & 4.530E+04 & 2.924E+39 & 240.80 \\ \hline
		UGC07690 & 1.18 & 54.581 & 1.02 & 3.641E+19 & 5.458E+04 & 2.029E+39 & 64.24 \\ \hline
		UGC07866 & 2.32 & 26.37 & 0.23 & 7.159E+19 & 2.637E+04 & 4.575E+38 & 266.10 \\ \hline
		UGC08286 & 4.25 & 75.61 & 4.28 & 1.311E+20 & 7.561E+04 & 8.513E+39 & 174.30 \\ \hline
		UGC08490 & 1.36 & 66.231 & 2.01 & 4.197E+19 & 6.623E+04 & 3.998E+39 & 47.64 \\ \hline
		UGC08550 & 1.47 & 47.479 & 0.74 & 4.536E+19 & 4.748E+04 & 1.472E+39 & 108.46 \\ \hline
		UGC08699 & 0.65 & 213.626 & 34.21 & 2.006E+19 & 2.136E+05 & 6.804E+40 & 4.07 \\ \hline
		UGC09037 & 13.41 & 105.162 & 26.45 & 4.138E+20 & 1.052E+05 & 5.261E+40 & 170.03 \\ \hline
		UGC09133 & 1.58 & 281.681 & 190.37 & 4.875E+19 & 2.817E+05 & 3.786E+41 & 2.35 \\ \hline
		UGC09992 & 3.12 & 33.537 & 0.54 & 9.627E+19 & 3.354E+04 & 1.074E+39 & 228.40 \\ \hline
		UGC10310 & 6.64 & 66.497 & 4.19 & 2.049E+20 & 6.650E+04 & 8.334E+39 & 265.58 \\ \hline
		UGC11455 & 14.56 & 216.144 & 161.99 & 4.493E+20 & 2.161E+05 & 3.222E+41 & 95.37 \\ \hline
		UGC11557 & 8.8 & 54.014 & 2.64 & 2.715E+20 & 5.401E+04 & 5.251E+39 & 511.52 \\ \hline
		UGC11914 & 3.28 & 291.722 & 121.53 & 1.012E+20 & 2.917E+05 & 2.417E+41 & 28.53 \\ \hline
		UGC12506 & 17.21 & 244.007 & 258.74 & 5.310E+20 & 2.440E+05 & 5.146E+41 & 84.82 \\ \hline
		UGC12632 & 6.4 & 60.031 & 4.48 & 1.975E+20 & 6.003E+04 & 8.911E+39 & 143.35 \\ \hline
		UGC12732 & 6.72 & 65.602 & 5.44 & 2.074E+20 & 6.560E+04 & 1.082E+40 & 152.86 \\ \hline
		UGCA281 & 0.74 & 24.303 & 0.31 & 2.283E+19 & 2.430E+04 & 6.166E+38 & 10.75 \\ \hline
		UGCA442 & 2.96 & 41.699 & 0.98 & 9.134E+19 & 4.170E+04 & 1.949E+39 & 149.18 \\ \hline
		UGCA444 & 2.33 & 25.915 & 0.18 & 7.190E+19 & 2.592E+04 & 3.580E+38 & 408.74 \\ \hline
	\end{longtable}

\end{document}